\documentclass[superscriptaddress,aps,pra,twocolumn,showpacs,nofootinbib,longbibliography]{revtex4-1}
\usepackage{etex}
\usepackage{amsmath,amssymb,amsthm}
\usepackage[colorlinks=true,citecolor=blue,urlcolor=blue]{hyperref}
\usepackage[pdftex]{graphicx}
\usepackage{times,txfonts}
\usepackage{braket}
\usepackage{color}
\usepackage{natbib}

\newcommand{\be}{\begin{equation}}
\newcommand{\ee}{\end{equation}}
\newcommand{\ba}{\begin{eqnarray}}
\newcommand{\ea}{\end{eqnarray}}
\newcommand{\ketbra}[2]{|#1\rangle \langle #2|}
\newcommand{\tr}{\operatorname{Tr}}

\newcommand{\etal}{{\it{et. al. }}}

\newtheorem{observation}{Observation}

\newtheorem{proposition}{Proposition}

\newtheorem{Lemma}{Lemma}

\newtheorem{scenario}{Scenario}
\begin{document}
\title{Tripartite entanglement detection through tripartite quantum steering in one-sided and two-sided 
device-independent scenarios}   

 \author{C. Jebaratnam}
 \email{jebarathinam@bose.res.in}
\affiliation{S. N. Bose National Centre for Basic Sciences, Block JD, Sector III, Salt Lake, Kolkata 700 098, India}   
   
\author{Debarshi Das}
\email{debarshidas@jcbose.ac.in}
\affiliation{Centre for Astroparticle Physics and Space Science (CAPSS), Bose Institute, Block EN, Sector V, Salt Lake, Kolkata 700 091, India}
 
\author{Arup Roy}
\affiliation{Physics and Applied Mathematics Unit, Indian Statistical Institute, 203 B. T. Road, Kolkata 700108, India.}

\author{Amit Mukherjee}
\affiliation{Optics $\&$ Quantum Information Group, The Institute of Mathematical Sciences,
HBNI, C.I.T Campus, Tharamani, Chennai 600 113, India}

\author{Some Sankar Bhattacharya}
\affiliation{Physics and Applied Mathematics Unit, Indian Statistical Institute, 203 B. T. Road, Kolkata 700108, India.}

\author{Bihalan Bhattacharya}
\affiliation{S. N. Bose National Centre for Basic Sciences, Block JD, Sector III, Salt Lake, Kolkata 700 098, India} 

\author{Alberto Riccardi}
\affiliation{Dip. Fisica and INFN Sez. Pavia, University of Pavia, via Bassi 6, I-27100 Pavia, Italy}

\author{Debasis  Sarkar}
\email{dsappmath@caluniv.ac.in}
\affiliation{Department of Applied Mathematics, University of Calcutta, 92, A.P.C. Road, Kolkata-700009, India.}

\date{\today}
\begin{abstract}
In the present work, we study tripartite quantum steering
of quantum correlations arising from two local dichotomic measurements on each side in the two types of partially device-independent scenarios: 
$1$-sided device-independent scenario where one of the parties performs untrusted measurements while the other two parties perform
trusted measurements and $2$-sided device-independent scenario where one of the parties performs trusted measurements while the other two parties perform untrusted measurements.
We demonstrate that tripartite steering in the $2$-sided device-independent scenario is weaker than 
tripartite steering in the $1$-sided device-independent scenario by using two families of quantum correlations. That is
these two families of quantum correlations in the $2$-sided device-independent framework detect tripartite entanglement through tripartite steering for a larger region than that in the $1$-sided device-independent framework.  It is shown that tripartite steering in the $2$-sided device-independent scenario implies the presence of genuine tripartite entanglement of $2\times 2 \times 2$ quantum system, even if the correlation does not exhibit genuine nonlocality or genuine steering. 
\end{abstract}


\maketitle

\section{Introduction}
 Multipartite entanglement is a resource for quantum information and computation when quantum networks  are considered. Therefore, detecting
 the presence of multipartite entanglement in quantum networks is an important problem in quantum information science. In particular,
 a genuinely multipartite entangled state (which is not separable with respect to any partitions)   \cite{guhne} is important  not only for quantum foundational research  but also in various  quantum information processing tasks, for example, in the context of extreme spin squeezing \cite{Sor}, high sensitive metrology tasks \cite{Hyl,toth}. Generation and detection of this 
kind of resource state is found to be difficult as the detection process deals with tomography and evaluation via constructing entanglement witness which require precise experimental control over the system subjected to measurements. But there is an alternative way to certify the presence of entanglement by observing the violation of Bell inequality \cite{bell} as entanglement is necessary ingredient to observe the violation. Motivated by this fact, a number of  multipartite Bell type inequalities \cite{SI, Nag, mb1, mb2, B, mb3} have been proposed to detect the genuine multipartite entanglement. To be specific, if the value of any Bell expression, in a Bell experiment, exceeds the value of the same expression obtained due to measurements on biseparable quantum states, then the presence of  genuine entanglement is guaranteed. This kind of research was first initiated in \cite{SI,Nag} but it took a shape by Bancal et. al. \cite{B} where they have constructed  device-independent entanglement witness (DIEW) of genuine multipartite entanglement for such Bell expressions.
 
The concept of quantum steering was first pointed out by Schrodinger \cite{scro} in the context of Einstein-Podolsky-Rosen
paradox (EPR) \cite{EPR}, which has no classical analogue.
Quantum steering as pointed out by Schrodinger occurs
when one of the two spatially separated observers
prepares genuinely different ensembles of quantum states for the other distant observer  by performing suitable quantum measurements 
on her/his side. 
Wiseman \etal \cite{steer} gave the formal definition of quantum steering from the foundational as well as quantum information perspective.
Quantum steering is certified by the violation of steering inequalities.
A number of steering inequalities have been proposed to observe steering \cite{ZHC16}. 
Violation of such steering inequalities certify the presence of entanglement in a one-sided device-independent way.

 In  Refs. \cite{UFNL,stm2}, the notion of steering has been generalized for multipartite scenarios  and multipartite steering 
 inequalities have been derived to detect multipartite entanglement in asymmetric networks where some of the parties' measurements are
 trusted while the other parties' measurements are uncharacterized.
 These studies did not examine genuine multipartite steering, in which the nonlocality, in the form of steering, is necessarily shared among all observers. Genuine multipartite steering has been proposed in \cite{stm3, stm4}. 
  In Refs. \cite{cava, stm6}, genuine tripartite steering inequalities have been derived to detect genuine tripartite entanglement 
  in a partially device-independent way.
 Characterization of multipartite quantum steering through semidefinite programming has also been performed \cite{cava, stm6, stm7}.
    
In the present work, we study  tripartite steering (which is analogous to standard Bell nonlocality) and genuine tripartite steering of quantum correlations arising from two local measurements on each side in the two types of partially device-independent scenarios: 
$1$-sided device-independent scenario where one of the parties performs untrusted measurements while the other two parties perform
trusted measurements and $2$-sided device-independent scenario where one of the parties performs trusted measurements while the other two parties perform untrusted measurements. 

In the $1$-sided device-independent framework, we study tripartite steering and genuine tripartite steering of two families of quantum correlations in the following scenarios: one of the parties performs two dichotomic black-box measurements and the other
two parties perform incompatible qubit measurements that
demonstrate Bell nonlocality \cite{CHS+69} in one of the types
or perform incompatible measurements that demonstrate EPR
steering without Bell nonlocality \cite{UFNL,jeba} in the other type. The first family of quantum correlation considered by us is called Svetlichny family as it can be obtained by performing the non-commuting measurements that lead to the violation of Svetlichny inequality and it violates Svetlichny inequality in a particular region. On the other hand, the second family of quantum correlation considered by us is called Mermin family as it can be obtained by performing the non-commuting measurements that
lead to the violation of Mermin inequality and it violates Mermin inequality in a particular region, but it does not violate Svetlichny inequality in any region. We demonstrate in which range these two families detect tripartite and genuine tripartite steering  in the aforementioned $1$SDI scenarios, respectively. 

We also explore in which range the Svetlichny family and Mermin family detect tripartite steering and genuine tripartite
steering in the $2$-sided device-independent framework. 

Our study demonstrates that tripartite steering in the $2$-sided device-independent framework is weaker than 
tripartite steering in the $1$-sided device-independent framework. In other words, tripartite steering in the context of 
$2$-sided device-independent framework detect tripartite entanglement for a larger region than that in the context of 
$1$-sided device-independent framework.
We demonstrate that tripartite steering in the $2$-sided device-independent scenario implies the presence of genuine tripartite entanglement of $2\times 2 \times 2$ quantum system, even if the correlation does not exhibit genuine nonlocality or genuine steering. 

The plan of the paper is as follows. In Sections \ref{Gsteer} and \ref{tes}
the fundamental ideas of tripartite nonlocality and that of tripartite
EPR steering in $1$-sided device-independent scenario as well as in $2$-sided device-independent scenario, respectively, are presented. In Sections
\ref{sts} and \ref{mts} tripartite steering and genuine tripartite steering in $1$-sided device-independent scenario as well as in $2$-sided device-independent scenario for Svetlichny family and Mermin family, respectively,
are discussed. Certifying genuine tripartite entanglement
of $2 \times 2 \times 2$ quantum system through tripartite steering inequality in $2$-sided device-independent scenario is also demonstrated in Sections
\ref{sts} and \ref{mts}. Finally, in the concluding Section \ref{Cnc}, we
discuss summary of the results obtained.

\section{Tripartite nonlocality}\label{Gsteer}
We consider a tripartite Bell scenario where three spatially separated parties, Alice, Bob and Charlie, 
perform two dichotomic measurements on their subsystems.
The correlation is described by the conditional probability distributions: $P(abc|A_xB_yC_z)$, here $x,y,z\in\{0,1\}$ and $a,b,c\in\{0,1\}$. 
The correlation exhibits standard tripartite nonlocality (i.e., Bell nonlocality) if it cannot be explained by a fully
local hidden variable (LHV) model,

\be
P(abc|A_xB_yC_z)=\sum_{\lambda} p_\lambda P_\lambda(a|A_x)P_\lambda(b|B_y)P_\lambda(c|C_z), \label{FLHV}
\ee

for some hidden variable $\lambda$ with probability distribution $p_\lambda$; $\sum_{\lambda} p_\lambda = 1$. The Mermin inequality (MI) \cite{mermin},
\begin{align}
\braket{M}:=\braket{A_0B_0C_1+A_0B_1C_0+A_1B_0C_0-A_1B_1C_1}_{LHV}\le2, \label{MI0}
\end{align}
is a Bell-type inequality whose violation implies that the correlation cannot be explained by a
fully local hidden variable  model as in Eq. (\ref{FLHV}). Here $\braket{A_xB_yC_z}=\sum_{abc}(-1)^{a \oplus b \oplus c} P(abc|A_xB_yC_z)$.

If a correlation violates a MI, it does not necessarily imply that it exhibits genuine tripartite nonlocality \cite{SI,B}.
In Ref. \cite{SI}, Svetlichny introduced the strongest form of genuine tripartite nonlocality 
(see Ref. \cite{B} for the other two forms 
of genuine nonlocality).
A correlation exhibits Svetlichny nonlocality if it cannot be explained by a hybrid nonlocal-LHV (NLHV) model,
\begin{align}
&P(abc|A_xB_yC_z)\!=\!\sum_\lambda p_\lambda P_\lambda(a|A_x)P_\lambda(bc|B_yC_z)+\nonumber \\
&\!\sum_\lambda q_\lambda P_\lambda(ac|A_xC_z)P_\lambda(b|B_y)\!+\!\sum_\lambda r_\lambda P_\lambda(ab|A_xB_y)P_\lambda(c|C_z), \label{HNLHV}
\end{align}
with $\sum_\lambda p_\lambda+\sum_\lambda q_\lambda+\sum_\lambda r_\lambda=1$. The bipartite probability distributions in this decomposition can have arbitrary nonlocality.

Svetlichny derived Bell-type inequalities to detect the strongest form of genuine  tripartite nonlocality \cite{SI}. For instance, one of the 
Svetlichny inequalities (SI) reads,
\ba
\braket{S}&:=&\braket{A_0B_0C_1+A_0B_1C_0+A_1B_0C_0-A_1B_1C_1}\nonumber \\
&&+\braket{A_0B_1C_1+A_1B_0C_1+A_1B_1C_0-A_0B_0C_0}\le4. \label{SI1}
\ea
Quantum correlations violate the SI up to $4\sqrt{2}$. A Greenberger-Horne-Zeilinger (GHZ)
state \cite{GHZ}  gives rise to the maximal violation of the SI for a different choice of measurements which do not demonstrate GHZ paradox \cite{mermin2}.

In the seminal paper \cite{mermin}, the MI was derived to demonstrate standard  tripartite nonlocality of 
three-qubit correlations arising from the genuinely entangled states. 
For this purpose, noncommuting measurements
that do not demonstrate Svetlichny nonlocality was used.
Note that when a Greenberger-Horne-Zeilinger (GHZ)
state \cite{GHZ} maximally violates the MI, the measurements
that give rise to it exhibit the GHZ paradox \cite{mermin2}.
\section{Definitions of tripartite EPR steering}\label{tes}
Before we define tripartite EPR steering, let us review the definition of bipartite EPR steering in the following $1$-sided device-independent
scenario. Two spatially separated parties, Alice (who is the trusted party) and Bob (who is the untrusted party)
share an unknown bipartite system described by the density matrix $\rho_{AB}$ in  $\mathbb{C}^{d_A} \otimes \mathbb{C}^{d_{B}}$ with 
the  dimension of Alice $d_A$ is known and the dimension of Bob $d_B$ is unknown. On this shared state, Bob performs black-box 
measurements (positive operator valued measurement, or in short, POVM) with the measurement operators $\{M_{b|y}\}_{b,y}$  ($M_{b|y} \geq 0$ $\forall b, y$; $\sum_{b} M_{b|y} = \mathbb{I}$ $\forall y$), here $y$ and $b$ denote the measurement choices and measurement outcomes 
of Bob, respectively, to prepare the set of conditional states on Alice's side.  The above steering scenario is characterized by the set 
of unnormalized conditional states on Alice's side $\{\sigma^{A}_{b|y}\}_{b,y}$, which is called an assemblage. Each element in this 
assemblage is given by $\sigma^{A}_{b|y}=\tr_B (\openone \otimes  M_{b|y} \rho_{AB})$.

Wiseman \etal \cite{steer} provided an operational definition of steering. According to this definition, Bob's measurements in the above scenario
demonstrates steerability to Alice iff the assemblage certifies entanglement. The assemblage which does not certify entanglement, i.e., does 
not imply steerability from Bob to Alice has a local hidden state (LHS) model as follows: for all $b$, $y$, each
element $\sigma_{b|y}^A$ in the assemblage admits the following decomposition:
\be
\sigma_{b|y}^A=\sum_\lambda q_\lambda P_\lambda(b|B_y) \rho^\lambda_A,
\ee
where $\lambda$ denotes classical random variable which occurs with probability 
$q_\lambda$; $\sum_{\lambda} q_\lambda = 1$; $P_\lambda(b|B_y)$ are some conditional probability distributions and the quantum states $\rho^\lambda_A$
are called local hidden states which satisfy $ \rho^\lambda_A \ge0$ and
$\tr \rho^\lambda_A=1$. Suppose Alice performs  positive operator valued measurements (POVM)  with measurement operators $\{M_{a|x}\}_{a,x}$ ($M_{a|x} \geq 0$ $\forall a, x$; $\sum_{a} M_{a|x} = \mathbb{I}$ $\forall x$)  on the
assemblage to detect steerability through the violation of a steering inequality. Then the scenario is characterized by the 
 set of conditional probability distributions,
 \be\label{biqs}
 P(ab|A_xB_y)=\tr\left( M_{a|x}    \sigma^A_{b|y} \right).
 \ee
 The above quantum correlation  $P(ab|A_xB_y)$ detects  steerability if and only if it cannot
 be explained by a  LHS-LHV  model of the form,
 \be
 P(ab|A_xB_y)=\sum_\lambda q_\lambda P(a|A_x,\rho_A^\lambda)P_\lambda(b|B_y) \hspace{0.5cm} \forall a, x, b, y, 
\ee
with $\sum_\lambda q_\lambda = 1$. Here $P(a|A_x,\rho_A^\lambda)$ 
are the distributions arising from the local hidden states 
$\rho_A^\lambda$. 

On the other hand, the quantum correlation $P(ab|A_xB_y)$ demonstrates Bell nonlocality if and only if
it cannot be explained by a  LHV-LHV  model of the form,
 \be
 P(ab|A_xB_y)=\sum_\lambda q_\lambda P_\lambda(a|A_x)P_\lambda(b|B_y) \hspace{0.5cm} \forall a, x, b, y, 
\ee
with $\sum_\lambda q_\lambda = 1$. The quantum correlation that does not have a LHV-LHV model also implies steering, on the other hand,
the quantum correlation that does not have a LHS-LHV model may not imply Bell nonlocality since certain local correlations may also
detect steering in the given $1$-sided device-independent scenario.

Let us now focus on the definition of tripartite steering. 
In the tripartite scenario, there are two types of partially device-independent scenarios 
where one can generalize bipartite EPR steering. These two scenarios
are called $1$-sided device-independent ($1$SDI) and $2$-sided device-independent ($2$SDI)  scenarios \cite{stm7}.
\subsection{Tripartite steering in $1$SDI scenario}
We will consider the following 
$1$-sided device-independent ($1$SDI) scenario  (depicted in FIG. \ref{ss}): 
Three spatially separated parties share an unknown tripartite quantum state $\rho^{ABC}$ in 
$\mathbb{C}^2 \otimes \mathbb{C}^{2} \otimes \mathbb{C}^{d}$  on which
Charlie performs black-box measurements (POVMs).
Suppose $M_{c|z}$ denote the unknown measurement operators of Charlie ($M_{c|z} \geq 0$ $\forall c, z$; $\sum_{c} M_{c|z} = \mathbb{I}$ $\forall z$). Then, the scenario is characterized by
 the set of (unnormalized) conditional two-qubit states on Alice and Bob's side $\{ \sigma^{AB}_{c|z} \}_{c,z}$, each element of which is given  as follows:
\be
\sigma^{AB}_{c|z}=\tr_C (\openone \otimes \openone \otimes M_{c|z} \rho^{ABC}).
\label{ass}
\ee
Alice and Bob can do local state tomography to determine the above assemblage prepared by Charlie.

 Analogous to the operational definition of bipartite EPR steering, we will now provide the operational 
 definition of tripartite steering in the above $1$SDI scenario.
 The assemblage $\sigma^{AB}_{c|z}$ given by Eq. (\ref{ass}) is called steerable   if\\
 
 \begin{figure}
\centering
\includegraphics[width=0.45\textwidth]{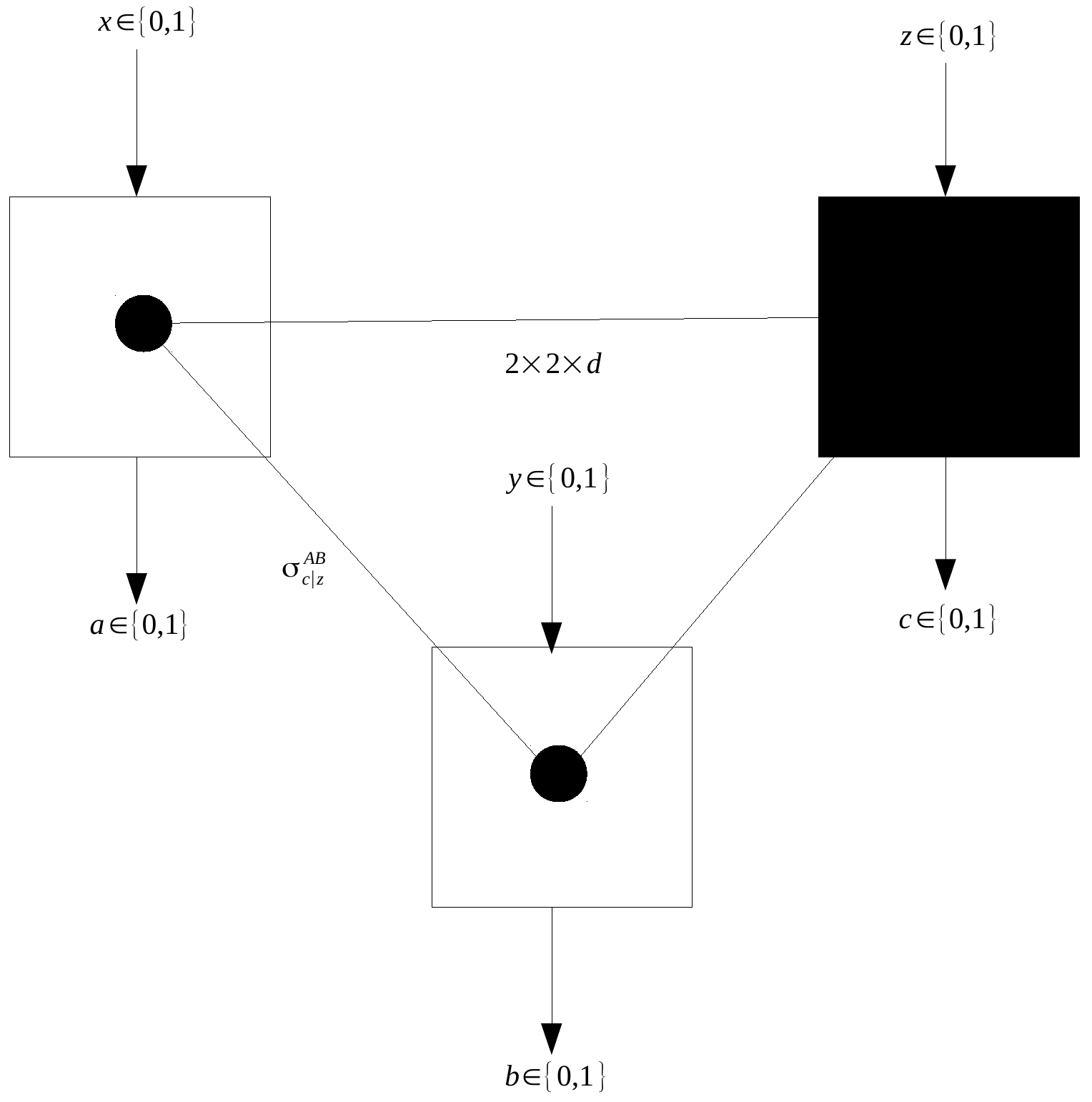}
\caption{Schematic diagram of our $1$SDI tripartite steering scenario: Alice, Bob and Charlie share a $2 \times 2 \times d$ quantum state. 
Charlie performs two dichotomic black-box measurements to produce assemblages $\sigma^{AB}_{c|z}$ (\ref{ass}) on Alice and Bob's side. 
On this assemblage, Alice and Bob perform two dichotomic measurements producing the joint probability distributions $P(abc|A_xB_yC_z)$ (here
$a,b,c$ denotes the outcomes and $x,y,z$ denotes the measurement choices)
to check whether Charlie demonstrates steerability to them through the violation of a steering inequality by $P(abc|A_xB_yC_z)$. In case of the scenario considered in Section \ref{sts}, Alice and Bob perform incompatible  qubit measurements that demonstrate
Bell nonlocality of certain two-qubit states \cite{CHS+69}; for instance, the singlet state.  On the other hand,  in case of the scenario considered in Section \ref{mts}, they perform incompatible qubit measurements that demonstrate EPR steering without Bell nonlocality of certain two-qubit states \cite{UFNL,jeba}; for instance, the singlet state.}\label{ss}
\end{figure}
 
  i)  the assemblage prepared on Alice and Bob's side cannot be reproduced
 by a fully separable state,  in $\mathbb{C}^2 \otimes \mathbb{C}^{2} \otimes \mathbb{C}^{d}$, of the form,
 \be
 \rho^{ABC}=\sum_\lambda p_\lambda \rho_A^\lambda \otimes \rho_B^\lambda \otimes \rho_C^\lambda,
 \ee
 with $\sum_\lambda p_\lambda = 1$; and\\
 
 ii)  entanglement between Charlie and Alice-Bob is detected.\\
 
 In the genuine steering scenario, Charlie 
 demonstrates genuine  tripartite  EPR steering 
 to Alice and Bob if the assemblage prepared on Alice and Bob's side cannot be reproduced
 by a biseparable state  in $\mathbb{C}^2 \otimes \mathbb{C}^{2} \otimes \mathbb{C}^{d}$,
 \be \label{bisep}
 \rho^{ABC}=\sum_\lambda p_\lambda \rho_A^\lambda \otimes \rho_{BC}^\lambda +
 \sum_\lambda q_\lambda \rho_{AC}^\lambda \otimes \rho_{B}^{\lambda} +\sum_\lambda r_\lambda \rho_{AB}^\lambda \otimes \rho_{C}^\lambda,
 \ee
 with $\sum_\lambda p_\lambda+\sum_\lambda q_\lambda + \sum_\lambda r_\lambda = 1$.

 Suppose in our tripartite $1$SDI scenario, the trusted parties Alice and Bob perform POVMs having elements $\{ M_{a|x} \}_{a,x}$ and $\{ M_{b|y} \}_{b,y}$, respectively, for detecting tripartite steering. Here $M_{a|x} \geq 0$ $\forall a, x$; $\sum_{a} M_{a|x} = \mathbb{I}$ $\forall x$; and $M_{b|y} \geq 0$ $\forall b, y$; $\sum_{b} M_{b|y} = \mathbb{I}$ $\forall y$. Then the scenario is characterized by the 
 set of conditional probability distributions,
 \be
 P(abc|A_xB_yC_z)=\tr\left( M_{a|x} \otimes M_{b|y}   \sigma^{AB}_{c|z} \right),
 \ee
 where $M_{a|x}$ and  $M_{b|y}$ are the measurement operators of Alice and Bob, respectively.
  Suppose the above quantum correlation  $P(abc|A_xB_yC_z)$ detects tripartite steerability. Then, it cannot
 be explained by  a fully LHS-LHV  model of the form,
 \be
 P(abc|A_xB_yC_z)=\sum_\lambda q_\lambda P(a|A_x,\rho_A^\lambda)P(b|B_y,\rho_B^\lambda)P_\lambda(c|C_z), \label{ALHS}
\ee
with $\sum_\lambda q_\lambda = 1$. Here $P(a|A_x,\rho_A^\lambda)$ 
and $P(b|B_y,\rho_B^\lambda)$ are the distributions arising from the local hidden states 
$\rho_A^\lambda$ and $\rho_B^\lambda$ which are in  $\mathbb{C}^{2}$, respectively. 
It should be noted that if a quantum correlation does not have an fully LHS-LHV model (\ref{ALHS}), then it does not necessarily 
imply that it detects tripartite steering from Charlie to Alice-Bob \cite{UFNL}.  The correlation $P(abc|A_xB_yC_z)$ detects tripartite steerability  if and only if \\
i) $P(abc|A_xB_yC_z)$ does not have a fully LHS-LHV  model as in Eq. (\ref{ALHS}); and \\
ii) entanglement between Charlie and Alice-Bob is detected.\\

The quantum correlation $P(abc|A_xB_yC_z)$ that detects  tripartite steering also
detects genuine tripartite steering if it cannot be explained by the following steering LHS-LHV (StLHS) model:
\begin{align}
P(abc|A_xB_yC_z)&=\!\sum_\lambda r_\lambda P(ab|A_xB_y,\rho^\lambda_{AB})P_\lambda(c|C_z) \nonumber \\
\!&+\!\sum_\lambda p_\lambda P(a|A_x,\rho^\lambda_A)P^Q_\lambda(bc|B_yC_z)\nonumber \\
&+\!\sum_\lambda q_\lambda P(b|B_y, \rho^\lambda_B) P^Q_\lambda(ac|A_xC_z), \label{NLHS}
\end{align}
with $\sum_\lambda p_\lambda+\sum_\lambda q_\lambda+\sum_\lambda r_\lambda=1$.  Here,
$P(a|A_x,\rho^\lambda_A)$ and $P(b|B_y, \rho^\lambda_B)$ are the distributions arising from
the qubit states $\rho^\lambda_A$ and $\rho^\lambda_B$ on Alice's side and Bob's side, respectively, $P_\lambda(c|C_z)$ is the distribution on Charlie's side arising from black-box measurements performed on a $d$ dimensional quantum state and $P^Q_\lambda(bc|B_yC_z)$ and $P^Q_\lambda(ac|A_xC_z)$ are the distributions that can be produced from a $2 \times d$ quantum states; and $P(ab|A_xB_y,\rho^\lambda_{AB})$ can be reproduced by two-qubit quantum states $\rho^\lambda_{AB}$ shared between Alice and Bob. Note that in the model given in Eq. (\ref{NLHS}), the bipartite distributions  at each $\lambda$ level may have Bell nonlocality or steering without Bell nonlocality \cite{UFNL,jeba}. Equivalently, the quantum correlation that detects genuine  tripartite steering cannot be reproduced by a biseparable
state in $\mathbb{C}^{2} \otimes \mathbb{C}^{2}\otimes \mathbb{C}^{d}$.

\subsection{Tripartite steering in $2$SDI scenario} \label{2SDID}

 \begin{figure}
\centering
\includegraphics[width=0.45\textwidth]{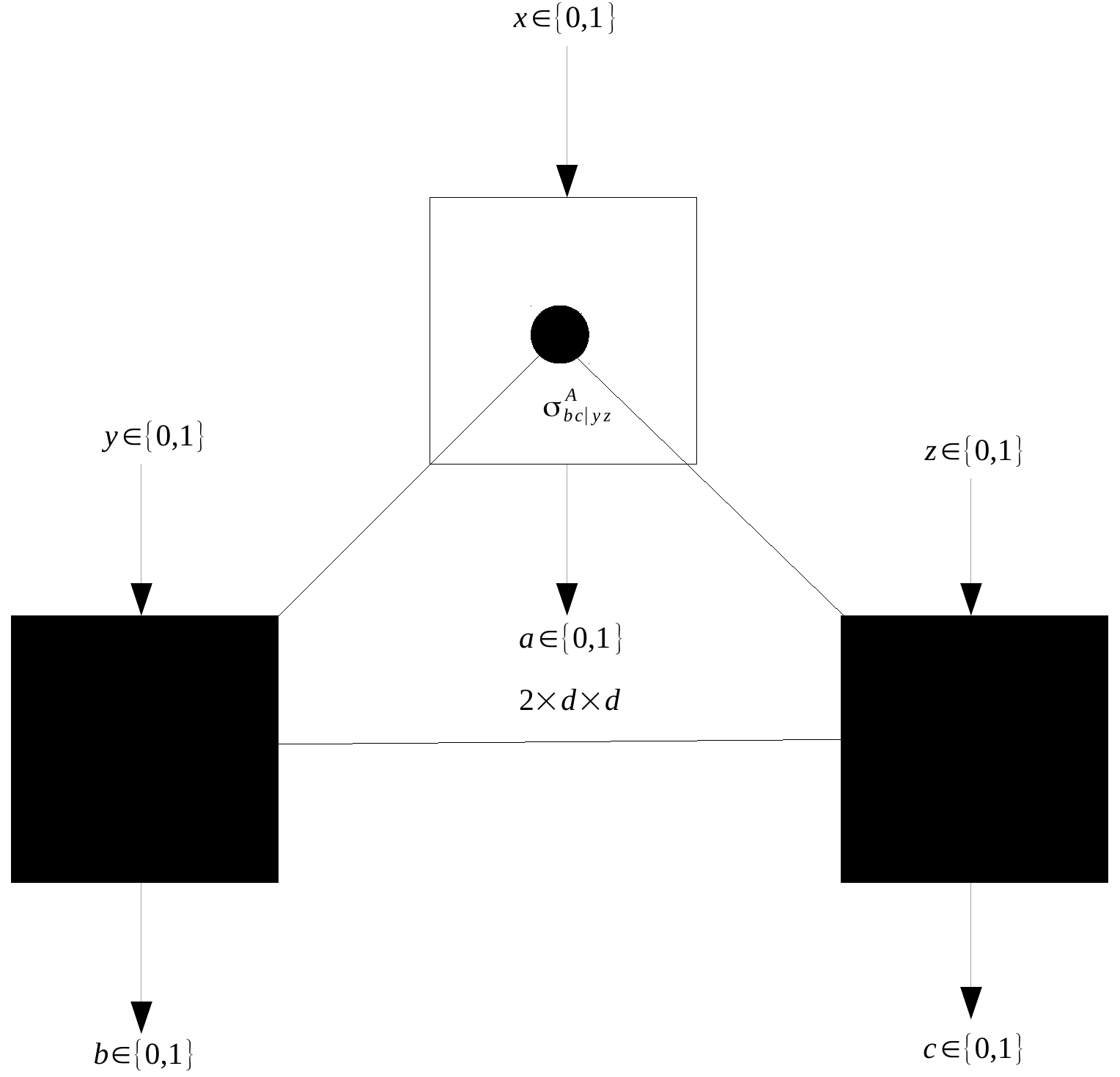}
\caption{Schematic diagram of our $2$SDI tripartite steering scenario: Alice, Bob and Charlie share a $2 \times d \times d$ quantum state. 
Bob and Charlie  perform two dichotomic black-box measurements  to produce assemblages $\sigma^{A}_{bc|yz}$ (\ref{ass1}) on Alice's side. 
On this assemblage, Alice  performs two dichotomic measurements producing the joint probability distributions $P(abc|A_xB_yC_z)$ (here
$a,b,c$ denotes the outcomes and $x,y,z$ denotes the measurement choices)
to check whether the assemblages $\sigma^{A}_{bc|yz}$ prepared by Bob and Charlie demonstrate steerability through the violation of a steering inequality by $P(abc|A_xB_yC_z)$.}\label{ss1}
\end{figure}

We will consider the following 
$2$-sided device-independent ($2$SDI) scenario  (depicted in FIG. \ref{ss1}): 
Three spatially separated parties share an unknown tripartite quantum state $\rho^{ABC}$ in 
$\mathbb{C}^2 \otimes \mathbb{C}^{d} \otimes \mathbb{C}^{d}$  on which
Bob and Charlie performs local black-box measurements  (POVMs).
Suppose $\{M_{b|y}\}_{b,y}$ and $\{M_{c|z}\}_{c,z}$ denote the unknown measurement operators of Bob and Charlie, respectively. Here $M_{b|y} \geq 0$ $\forall b, y$; $\sum_{b} M_{b|y} = \mathbb{I}$ $\forall y$ and $M_{c|z} \geq 0$ $\forall c, z$; $\sum_{c} M_{c|z} = \mathbb{I}$ $\forall z$. 
Then, the scenario is characterized by
 the set of (unnormalized) conditional qubit states on Alice's side $\{\sigma^{A}_{bc|yz}\}_{b,c,y,z}$. 
 The each element in this assemblage is given  as follows:
\be
\sigma^{A}_{bc|yz}=\tr_{BC} (\openone \otimes M_{b|y} \otimes M_{c|z} \rho^{ABC}).
\label{ass1}
\ee
Alice can do local state tomography to determine the above assemblage prepared by Charlie.

 We will now provide the operational 
 definition of tripartite steering in the above $2$SDI scenario. 
 The assemblage $\{\sigma^{A}_{bc|yz}\}_{b,c,yz}$ is called steerable  if it cannot be reproduced
 by a fully separable state   in 
$\mathbb{C}^2 \otimes \mathbb{C}^{d} \otimes \mathbb{C}^{d}$  of the form,
 \be
 \rho^{ABC}=\sum_\lambda p_\lambda \rho_A^\lambda \otimes \rho_B^\lambda \otimes \rho_C^\lambda,
 \ee
 with $\sum_\lambda p_\lambda = 1$ in the given steering scenario.
In our $2$SDI scenario, even if entanglement is not certified between Alice and Bob-Charlie, tripartite steering can still
 occur by the presence of Bell nonlocality between Charlie and Bob \cite{UFNL}. 
 When entanglement between Alice and Bob-Charlie is detected, our $2$SDI scenario demonstrates genuine tripartite steering
 if the assemblage $\sigma^{A}_{bc|yz}$  cannot be reproduced
 by a biseparable state as given by Eq. (\ref{bisep}) in $\mathbb{C}^2 \otimes \mathbb{C}^{d} \otimes \mathbb{C}^{d}$.

 Suppose in our tripartite $2$SDI scenario, the trusted party Alice performs  POVMs having elements $\{ M_{a|x} \}_{a,x}$ for detecting tripartite steering. Here $M_{a|x} \geq 0$ $\forall a, x$; $\sum_{a} M_{a|x} = \mathbb{I}$ $\forall x$. Then the scenario is characterized by the 
 set of conditional probability distributions,
 \be
 P(abc|A_xB_yC_z)=\tr\left( M_{a|x}    \sigma^{A}_{bc|yz} \right),
 \ee
 where $M_{a|x}$  are the measurement operators of Alice.
  Suppose the above quantum correlation  $P(abc|A_xB_yC_z)$ cannot
 be explained by a fully LHS-LHV  model of the form, 
 \be
 P(abc|A_xB_yC_z)=\sum_\lambda q_\lambda P(a|A_x,\rho_A^\lambda)P_\lambda(b|B_y)P_\lambda(c|C_z), \label{ALHS1}
\ee
with $\sum_\lambda q_\lambda = 1$ (Here, $P(a|A_x,\rho_A^\lambda)$ 
are the distributions arising from the local hidden states 
$\rho_A^\lambda$  which are in  $\mathbb{C}^{2}$). Then, it detects tripartite steerability. 

The quantum correlation $P(abc|A_xB_yC_z)$ that detects tripartite steering in our $2$SDI scenario also
detects genuine tripartite steering if it cannot be explained by the following steering LHS-LHV (StLHS) model:
\begin{align}
P(abc|A_xB_yC_z)&=\!\sum_\lambda r_\lambda P^Q_\lambda(ab|A_xB_y)P_\lambda(c|C_z) \nonumber \\
\!&+\!\sum_\lambda p_\lambda P(a|A_x,\rho^\lambda_A)P_\lambda(bc|B_yC_z)\nonumber \\
&+\!\sum_\lambda q_\lambda P_\lambda(b|B_y) P^Q_\lambda(ac|A_xC_z), \label{NLHS1}
\end{align}
with $\sum_\lambda p_\lambda+\sum_\lambda q_\lambda+\sum_\lambda r_\lambda=1$.  Here,
$P(a|A_x,\rho^\lambda_A)$ are the distributions arising from
the qubit states $\rho^\lambda_A$ and, $P_\lambda(b|B_y)$ and $P_\lambda(c|C_z)$ are the distribution on Bob's and Charlie's sides, respectively, arising from black-box measurements performed on a $d$ dimensional quantum state and $P^Q_\lambda(ab|A_xB_y)$ and $P^Q_\lambda(ac|A_xC_z)$ are the 
distribution that can be produced from a $2 \times d$ quantum state; and $P_\lambda(bc|B_y C_z)$ can be reproduced by a $d \times d$ quantum state. Note that in the model given in Eq. (\ref{NLHS1}), the bipartite distributions  at each $\lambda$ level may have Bell nonlocality or steering without Bell nonlocality \cite{UFNL,jeba}. Equivalently, the quantum correlation that detects genuine  tripartite steering in our $2$SDI cannot be reproduced by a biseparable state in $\mathbb{C}^{2} \otimes \mathbb{C}^{d}\otimes \mathbb{C}^{d}$.

 In the next Section we study in which range two one-parameter families of quantum correlations obtained from local dichotomic measurements on tripartite quantum states  detect tripartite quantum steering in our $1$SDI and $2$SDI scenarios. 

%\section{Detection of tripartite steering with two families of quantum correlations}
%Here we study in which range two one-parameter families of quantum correlations obtained from local dichotomic measurements on tripartite quantum states  detect tripartite quantum steering in our $1$SDI and $2$SDI scenarios. 
%{\color{red} These two families can be obtained from
%noisy three-qubit GHZ state,
%$\rho=V\ketbra{\Phi_{GHZ}}{\Phi_{GHZ}}+(1-V)\openone/8$, where $\ket{\Phi_{GHZ}}=\frac{1}{\sqrt{2}}(\ket{000}+\ket{111})$,
%by performing the non-commuting measurements that lead to violation of the Svetlichny and Mermin inequalities, respectively. For this reason, these two families are called Svetlichny and Mermin family,
%respectively. The noisy three-qubit GHZ state is genuinely entangled iff $V>0.429$ \cite{Guhne}.}

\section{Detection of tripartite steering with Svetlichny family} \label{sts}
The Svetlichny family of tripartite correlations is defined as:
\be
P_{SvF}^{V} (abc|A_x B_ y C_z) = \frac{2 + (-1)^{a \oplus b \oplus c \oplus  xy \oplus yz \oplus xz} \sqrt{2} V}{16},
\label{SvF}
\ee
where $0 \le V\le 1$, which can be obtained from the noisy three-qubit GHZ state,
$\rho=V\ketbra{\Phi_{GHZ}}{\Phi_{GHZ}}+(1-V)\openone/8$, where $\ket{\Phi_{GHZ}}=\frac{1}{\sqrt{2}}(\ket{000}+\ket{111})$, for the measurements that 
give rise to the maximal violation of the SI; for instance,
$A_0=\sigma_x$, $A_1=\sigma_y$, $B_0=\frac{\sigma_x-\sigma_y}{\sqrt{2}}$, $B_1=\frac{\sigma_x + \sigma_y}{\sqrt{2}}$, 
$C_0=\sigma_x$ and $C_1=\sigma_y$.  The noisy three-qubit GHZ state is genuinely entangled iff $V>0.429$ \cite{Guhne}.
The Svetlichny family certifies genuine entanglement in a fully device independent way for $V>\frac{1}{\sqrt{2}}$,
as it violates the SI in this range. 
The Svetlichny family  has a fully local hidden variable (LHV) model
when $V\le \frac{1}{\sqrt{2}}$ \cite{B}. This implies that
in this range, it can also arise from a separable state in the higher dimensional space \cite{DQKD}.

\subsection{$1$SDI scenario}
We consider a tripartite  $1$SDI  steering scenario where Charlie performs two dichotomic black-box measurements to prepare conditional two-qubit states on  Alice and Bob's side on which 
Alice and Bob perform pair of incompatible qubit  measurements that demonstrate Bell nonlocality of certain two-qubit states; for instance, the singlet state. Now we are going to present a Lemma which is useful to find out in which ranges the Svetlichny family detects genuine tripartite steering and tripartite steering in the context of the above $1$SDI scenario.
\begin{Lemma}\label{gscssf}
In our $1$SDI scenario mentioned above the Svetlichny family has a steering LHS-LHV model as in Eq. (\ref{NLHS})  in the range $0 < V \le\frac{1}{\sqrt{2}}$
 and has a fully LHS-LHV model as in Eq. (\ref{ALHS}) 
 iff $0 < V\le\frac{1}{2\sqrt{2}}$.
\end{Lemma}
\begin{proof}
See Appendix \ref{PSLHVSF}.
\end{proof}

The above Lemma implies the following two propositions.
\begin{proposition}
The Svetlichny family   detects genuine tripartite steering 
iff $V>\frac{1}{\sqrt{2}}$ in the context of our $1$SDI scenario.
\label{ppnp1}
\end{proposition}
\begin{proof}
Since the Svetlichny family violates the Svetlichny inequality for $V>\frac{1}{\sqrt{2}}$, it certifies genuine tripartite entanglement in a fully device independent way in that range. Hence, it is followed that the Svetlichny family certifies genuine tripartite entanglement in our $1$SDI scenario as well for $V>\frac{1}{\sqrt{2}}$. The Svetlichny family, therefore, does not have a steering  LHS-LHV model as in Eq.(\ref{NLHS}) in our $1$SDI scenario for $V>\frac{1}{\sqrt{2}}$. On the other hand, following Lemma \ref{gscssf} we can state that the Svetlichny family has a steering LHS-LHV model as in Eq. (\ref{NLHS}) in our $1$SDI scenario in the range $0 < V \le\frac{1}{\sqrt{2}}$. Hence, the Svetlichny family   detects genuine tripartite steering 
iff $V>\frac{1}{\sqrt{2}}$ in the context of our $1$SDI scenario.
\end{proof}

\begin{proposition}
The Svetlichny family detects  tripartite steering 
iff $V>\frac{1}{\sqrt{2}}$ in the context of our $1$SDI scenario. 
\label{p2}
\end{proposition}
\begin{proof}
Svetlichny family detects entanglement between Charlie and Alice-Bob
for $V > \frac{1}{\sqrt{2}}$ as it violates the Svetlichny inequality in this range. Moreover, the steering LHS-LHV model given in the proof of Lemma \ref{gscssf} for the Svetlichny family implies that for $V\le \frac{1}{\sqrt{2}}$, it
can be reproduced by a $2 \times 2 \times d$ biseparable state of the 
form,
 \be \label{bisep1}
 \rho^{ABC}=\sum^3_{\lambda=0} r_\lambda \rho_{AB}^\lambda \otimes \ketbra{\lambda}{\lambda},
 \ee
 with $\sum_\lambda r_\lambda = 1$. Therefore, the Svetlichny family detects entanglement between Charlie and Alice-Bob
 iff $V > \frac{1}{\sqrt{2}}$.
On the other hand, in the context of our $1$SDI scenario, the Svetlichny family does not have a fully LHS-LHV model as in Eq. (\ref{ALHS}) following Lemma \ref{gscssf} 
 for $V > \frac{1}{2\sqrt{2}}$. Combining these two facts we can state that the Svetlichny family detects entanglement between Charlie and Alice-Bob and does not have a fully LHS-LHV model as in Eq. (\ref{ALHS}) in the range $V>\frac{1}{\sqrt{2}}$ following Lemma \ref{gscssf}. Hence, in the context of our $1$SDI scenario,
the Svetlichny family   detects  tripartite steering 
iff $V>\frac{1}{\sqrt{2}}$. 
\end{proof}

From the Propositions \ref{ppnp1} and \ref{p2} we observe the following two salient features: 1) in our 1SDI scenario, the Svetlichny family does not detect tripartite steering in the range $\frac{1}{2 \sqrt{2}}  < V \leq \frac{1}{\sqrt{2}}$ despite it does not have a fully LHS-LHV model in this range and 2) 
the ranges in which the Svetlichny family detects tripartite steering and genuine tripartite steering in our 1SDI scenario are the same.

\subsection{$2$SDI scenario}
We now consider a tripartite $2$SDI steering scenario where Bob and Charlie perform two dichotomic black-box measurements to prepare conditional single qubit states on 
Alice's side  on which Alice performs two mutually unbiased  qubit  measurements.
We are now interested in which ranges the Svetlichny family detects genuine tripartite steering and tripartite steering in the context of this $2$SDI scenario.

\begin{proposition}
 The Svetlichny family detects genuine tripartite steering in our $2$SDI scenario iff $V>\frac{1}{\sqrt{2}}$.
\end{proposition}
\begin{proof}
 Note that the Svetlichny family can be reproduced by a $2 \times 2 \times d$ dimensional biseparable state of the form given in Eq. (\ref{bisep1}) for $V\le \frac{1}{\sqrt{2}}$.
 This implies that it  does not detect genuine tripartite entanglement in the range $ V\le \frac{1}{\sqrt{2}}$ in our $2$SDI scenario. On the other hand, the Svetlichny family detects genuine tripartite entanglement  for $V>\frac{1}{\sqrt{2}}$ in the fully device independent scenario as it violates the Svetlichny inequality in this range. Hence, the Svetlichny family detects genuine tripartite entanglement  for $V>\frac{1}{\sqrt{2}}$ in our $2$SDI scenario as well. The Svetlichny family, therefore, detects genuine tripartite steering in our $2$SDI scenario iff $V>\frac{1}{\sqrt{2}}$.
  \end{proof}
 
 \begin{proposition}
  The Svetlichny family detects tripartite steering in our $2$SDI scenario for $V> \frac{1}{2}$.
  \label{p4}
 \end{proposition}
\begin{proof}
In Ref. \cite{UFNL}, it has been shown that the violation of the following inequality (Eq. (22) in \cite{UFNL} with $N$ (Number of parties) $=3$ and $T$ (Number of trusted parties) $=1$):
\ba \label{SIAEPR}
\braket{S}_{2\times ? \times ?} \overset{\mathrm{LHS}}{\le} 2\sqrt{2},
\ea
detects tripartite steering in our $2$SDI scenario. Here, $S$ is the Svetlichny operator given in the Svetlichny inequality (\ref{SI1}), $2\times ? \times ?$ indicates that Alice performs qubit measurements
while Bob and Charlie perform black-box measurements.
Note that the Svetlichny family  
violates the above steering inequality for $V>\frac{1}{2}$. Thus, the Svetlichny family detects tripartite steering
for $V>\frac{1}{2}$ in the $2$SDI scenario. 
\end{proof}

From the aforementioned Propositions we observe the following two salient features:
1) the ranges in which the Svetlichny family detects tripartite steering and genuine tripartite steering in our $2$SDI 
scenario are different and 2) Svetlichny family detects more tripartite entangled states to be tripartite steerable in the $2$SDI scenario than
in the case of $1$SDI scenario. Now we are going to make the following important observation.

\begin{observation}
Quantum violation of the tripartite steering inequality (\ref{SIAEPR}) by $2 \times 2 \times 2$ systems  certifies genuine entanglement in that $2 \times 2 \times 2$ systems, even if genuine nonlocality or genuine steering is not detected.
\end{observation}
\begin{proof}
We consider the following Svetlichny biseparability  inequality: 
\be
\braket{S}_{2\times 2 \times 2} \overset{\mathrm{Bi-sep}}{\le} 2\sqrt{2}, 
\label{deeqbi}
\ee
whose violation detects genuine tripartite entanglement in $2 \times 2 \times 2$ systems (for derivation see the appendix \ref{bisepapp}).
Here, $\braket{S}_{2\times 2 \times 2}$ denotes the Svetlichny operator with the measurement observables on each side being incompatible qubit
measurements.
Note that quantum violation of tripartite steering inequality (\ref{SIAEPR}) by $2 \times 2 \times 2$ systems implies quantum violation of the Svetlichny biseparability  inequality (\ref{deeqbi})
by that $2 \times 2 \times 2$ systems.  Because, for both of these two inequalities the upper bounds are the same. Hence the claim.
\end{proof} 

We have illustrated the above results with the Svetlichny family in Fig. \ref{SvFF}.

\begin{figure}
\centering
\includegraphics[width=0.45\textwidth]{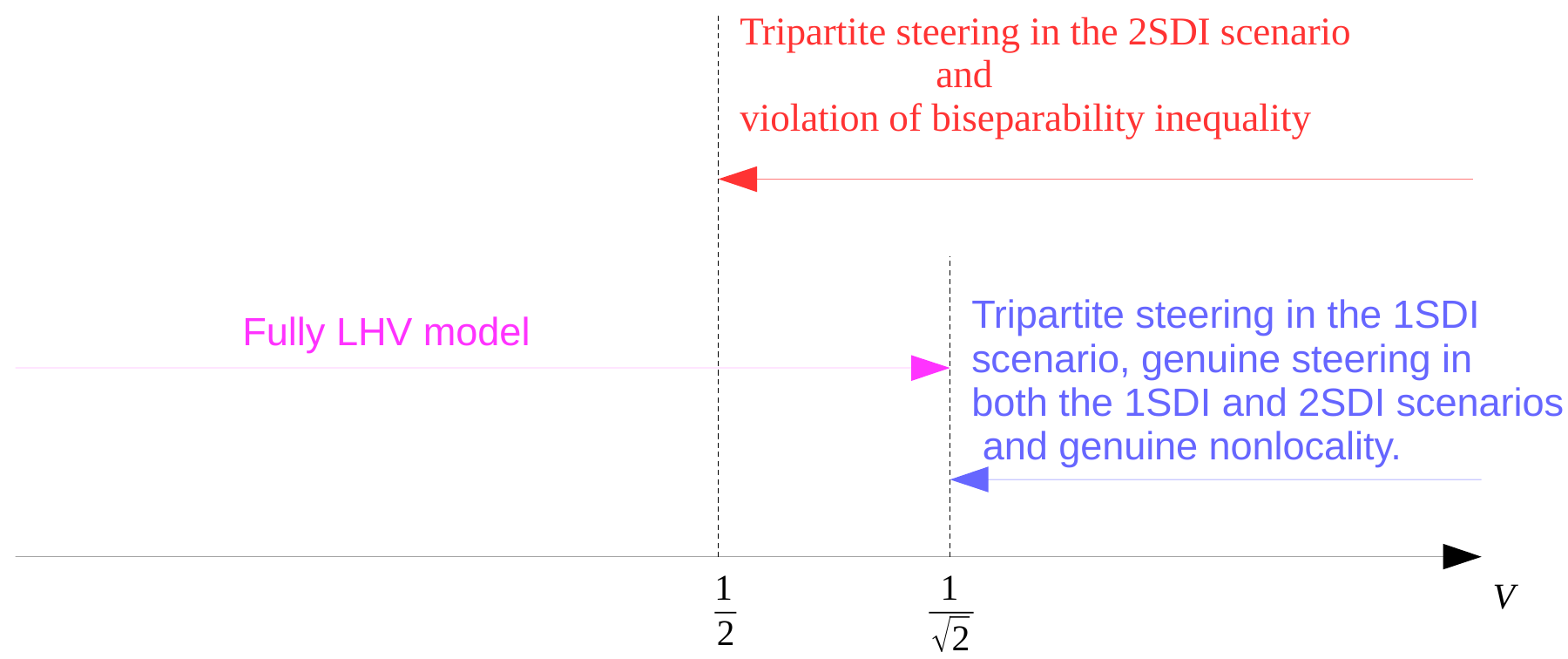}
\caption{Regions of the parameter $V$ in which the Svetlichny family is genuinely nonlocal, detects genuine steering and tripartite
steering, has a fully LHV model and violates biseparability inequality.}\label{SvFF}
\end{figure}

\section{Detection of tripartite steering with Mermin family}\label{mts}
The Mermin family of tripartite correlations is defined  as
\begin{equation}
P_{MF}^{V} (abc|A_x B_ y C_z) = \frac{1 + (-1)^{a \oplus b \oplus c \oplus  xy \oplus yz \oplus xz} \delta_{x \oplus y \oplus 1,z} V}{8},  
\label{mfameq}
\end{equation}
where $0 < V \leq 1$, which can be obtained from the noisy three-qubit GHZ state for the measurements 
that give rise to the GHZ paradox; for instance, $A_0=\sigma_x$, $A_1=\sigma_y$, $B_0=\sigma_x$, $B_1=\sigma_y$,
$C_0=\sigma_x$ and $C_1=-\sigma_y$.
The Mermin family is Bell nonlocal for $V>\frac{1}{2}$ as it violates the MI given in Eq. (\ref{MI0}).
This implies that it certifies tripartite entanglement for $V>\frac{1}{2}$. In that range, 
the Mermin family is not genuinely nonlocal since it has a NLHV model as in Eq. (\ref{HNLHV}) \cite{JGA17}. However, it certifies
genuine tripartite entanglement for $V>\frac{1}{\sqrt{2}}$ in a fully device-independent way since it violates the Mermin inequality more than $2\sqrt{2}$ \cite{SU01}.
We will study tripartite steering of the Mermin family in our $1$SDI and $2$SDI scenarios.

\begin{figure}
\centering
\includegraphics[width=0.45\textwidth]{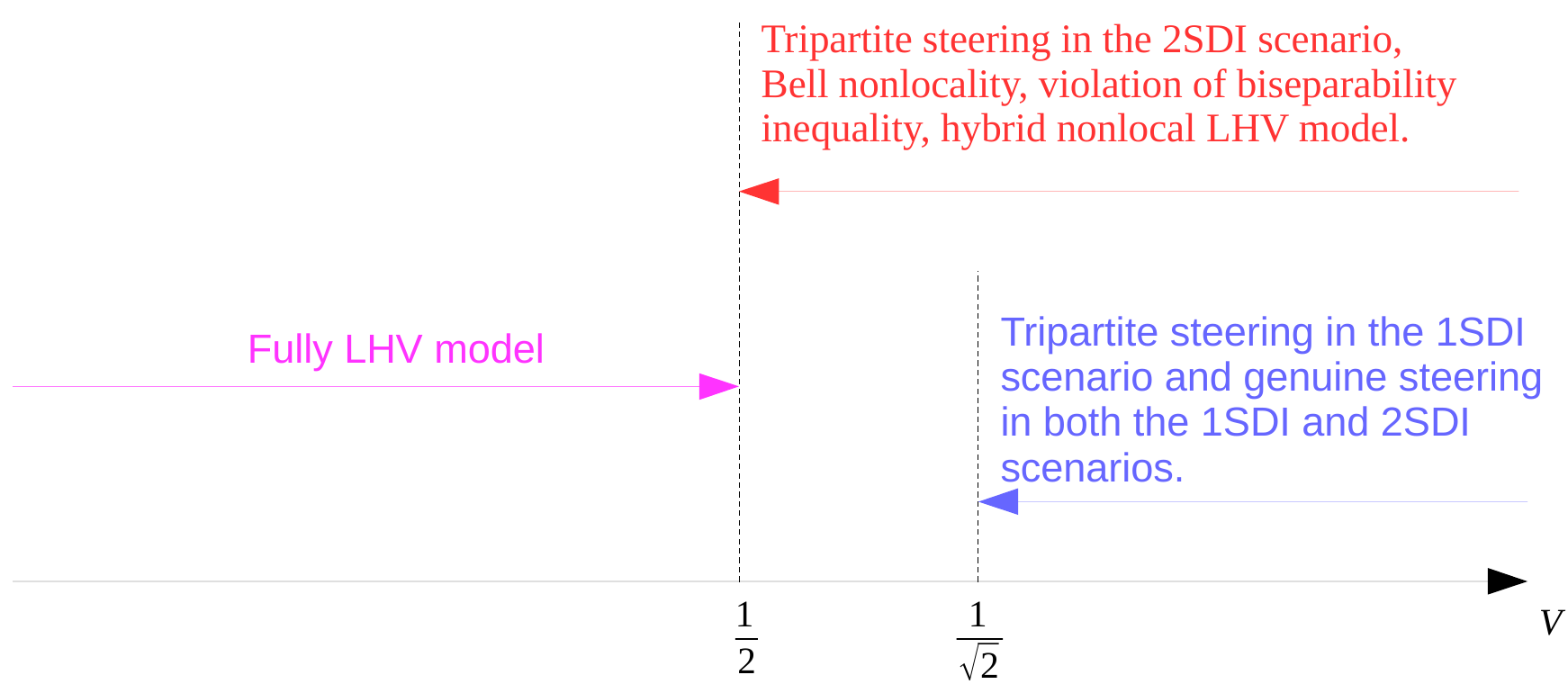}
\caption{Regions of the parameter $V$ in which the Mermin family detects Bell-nonlocality, genuine steering and tripartite steering, has fully LHV model and hybrid nonlocal LHV model, and violates the biseparability inequality.}\label{MeFF}
\end{figure}

\subsection{$1$SDI scenario}
We consider a tripartite  $1$SDI steering scenario where Charlie performs two dichotomic black-box measurements to prepare conditional two-qubit states on  Alice and Bob's side on which 
Alice and Bob perform pair of incompatible qubit  measurements that demonstrate EPR steering without Bell nonlocality.
Now we present a Lemma which is useful to find out in which ranges the Mermin family detects genuine tripartite steering and tripartite steering in the context of the above $1$SDI scenario.
\begin{Lemma}\label{gscsmf}
In our $1$SDI scenario mentioned above the Mermin family has a steering LHS-LHV model as in Eq. (\ref{NLHS}) in the range $0 < V \le\frac{1}{\sqrt{2}}$ and
has a fully LHS-LHV model iff $0 < V\leq\frac{1}{2\sqrt{2}}$.
\end{Lemma}

\begin{proof}
See Appendix \ref{PSLHVMF} for the proof.
\end{proof}

The above Lemma implies the following two propositions.
\begin{proposition}
The Mermin family   detects genuine tripartite steering 
iff $V>\frac{1}{\sqrt{2}}$ in the context of $1$SDI scenario.
\label{ppnp3}
\end{proposition}
\begin{proof}
 Since the Mermin family violates the Mermin inequality more than $2\sqrt{2}$ for $V>\frac{1}{\sqrt{2}}$, it certifies genuine tripartite entanglement 
 in the fully device independent scenario in that range \cite{SU01}.  Hence, the Mermin family certifies genuine tripartite entanglement in our $1$SDI scenario as well for $V>\frac{1}{\sqrt{2}}$. This implies that, for $V>\frac{1}{\sqrt{2}}$, it does not have a steering LHS-LHV model as in Eq. (\ref{NLHS}) in our $1$SDI scenario. On the other hand, following Lemma \ref{gscsmf} we can state that the Mermin family has a steering LHS-LHV model as in Eq. (\ref{NLHS}) in the range $0 < V \le\frac{1}{\sqrt{2}}$. The Mermin family, therefore,  detects genuine tripartite steering 
iff $V>\frac{1}{\sqrt{2}}$.
\end{proof}

\begin{proposition}
The Mermin family  detects  tripartite steering 
iff $V>\frac{1}{\sqrt{2}}$ in the context of our $1$SDI scenario. 
\label{ppnp4}
\end{proposition}
\begin{proof}
Mermin family detects entanglement between Charlie and Alice-Bob
for $V > \frac{1}{\sqrt{2}}$ as it violates the Mermin inequality more than $2\sqrt{2}$  in this range \cite{SU01}.  Moreover, the steering LHS-LHV model given in the proof of Lemma \ref{gscsmf} for the Mermin family implies that, for $V\le \frac{1}{\sqrt{2}}$, it
can be reproduced by a $2 \times 2 \times d$ biseparable state of the 
form given by Eq. (\ref{bisep1}). Therefore, the Mermin family detects entanglement between Charlie and Alice-Bob
 iff $V > \frac{1}{\sqrt{2}}$.
On the other hand, in the context of our $1$SDI scenario, the Mermin family does not have a fully LHS-LHV model as in Eq. (\ref{ALHS}) 
 for $V > \frac{1}{2\sqrt{2}}$ following Lemma \ref{gscsmf}. Combining these two facts we can state that the Mermin family detects entanglement between Charlie and Alice-Bob and does not have a fully LHS-LHV model as in Eq. (\ref{ALHS}) in the range $V>\frac{1}{\sqrt{2}}$ in our $1$SDI scenario. Hence, in the context of our $1$SDI scenario,
the Mermin family   detects  tripartite steering 
iff $V>\frac{1}{\sqrt{2}}$. 
\end{proof}

From the Propositions \ref{ppnp3} and \ref{ppnp4}, we observe the following two salient features: $1$) in our 1SDI scenario, the Mermin family does not detect tripartite steering in the range $\frac{1}{2 \sqrt{2}}  < V \leq \frac{1}{\sqrt{2}}$ despite it does not have a fully LHS-LHV model in this range and $2$) the ranges in which the Mermin family detects tripartite steering and genuine tripartite steering in our 1SDI scenario are the same.

\subsection{$2$SDI scenario}
We will now study tripartite steering of the Mermin family in our $2$SDI scenario where Bob and Charlie perform two dichotomic black-box measurements to prepare conditional single qubit states on 
Alice's side  on which Alice performs two mutually unbiased  qubit  measurements.
We are interested to find out in which ranges the Mermin family detects genuine  tripartite steering and tripartite steering in the context of this $2$SDI scenario.
\begin{proposition}
 The Mermin family detects genuine tripartite steering in our $2$SDI scenario iff $V>\frac{1}{\sqrt{2}}$.
\end{proposition}
\begin{proof}
 Note that the Mermin family can be reproduced by a $2 \times 2 \times d$ dimensional biseparable state of the form given in Eq. (\ref{bisep1}) for $V\le \frac{1}{\sqrt{2}}$.
 This implies that it  does not detect genuine tripartite entanglement in the range $V\le \frac{1}{\sqrt{2}}$ in our $2$SDI scenario. On the other hand, the Mermin family detects genuine tripartite entanglement for $V>\frac{1}{\sqrt{2}}$ in the fully device independent scenario as the Mermin family violates the Mermin inequality more than $2\sqrt{2}$ in that range \cite{SU01}. Hence, the Mermin family detects genuine tripartite entanglement in our $2$SDI scenario as well for $V>\frac{1}{\sqrt{2}}$. The Mermin family, therefore, detects 
 genuine tripartite steering in our $2$SDI scenario iff $V>\frac{1}{\sqrt{2}}$.
 \end{proof}
 
 \begin{proposition}
   The Mermin family detects  tripartite steering in our $2$SDI scenario for $V> \frac{1}{2}$.
 \end{proposition}
\begin{proof}
 In Ref. \cite{UFNL}, it has been shown that the violation of the following inequality (Eq. ($21$) in \cite{UFNL} with $N=3$ and $T=1$),
\ba \label{MIAEPR}
\braket{M}_{2\times ? \times ?} \overset{\mathrm{LHS}}{\le} 2,
\ea
detects tripartite steering in our $2$SDI scenario. Here, $M$ is the Mermin operator given in the Mermin inequality (\ref{MI0}), $2\times ? \times ?$ indicates that Alice performs qubit measurements
while Bob and Charlie perform black-box measurements. Note that the Mermin family  
violates the above steering inequality for $V> \frac{1}{2}$. Thus, the Mermin family detects tripartite steering
for $V>\frac{1}{2}$ in the $2$SDI scenario.
\end{proof}

From aforementioned Propositions we observe the following two salient features:
1) the ranges in which the Mermin family detects tripartite steering and genuine tripartite steering in our $2$SDI 
scenario are different and 2) the Mermin family detects more tripartite entangled states to be tripartite steerable in the $2$SDI scenario than
in the case of $1$SDI scenario. Now we want to state the following important observation.

\begin{observation}
Quantum violation of the  tripartite steering inequality (\ref{MIAEPR}) by $2 \times 2 \times 2$ systems  certifies genuine entanglement in that $2 \times 2 \times 2$ systems, even if genuine nonlocality or genuine steering is not detected.
\end{observation}
\begin{proof}
In Ref. \cite{TGS05}, it was shown that the Mermin inequality detect genuine entanglement of three-qubit systems in the scenario where all three parties perform two mutually unbiased qubit measurements.
This implies that the violation of the inequality (\ref{MIAEPR}) implies the presence of 
genuine entanglement if all three parties performs  qubit measurements in mutually unbiased bases.
Similar to the derivation of  Svetlichny biseparability inequality presented in Appendix \ref{bisepapp}, 
one can obtain following the Mermin biseparability inequality:
\be
\braket{M}_{2\times 2 \times 2} \overset{\mathrm{Bi-sep}}{\le} 2, 
\label{bosepnewequ}
\ee
whose violation detects genuine tripartite entanglement of $ 2 \times 2 \times 2$ systems.
Here, $\braket{M}_{2\times 2 \times 2}$ denotes the Mermin operator with the measurement observables on each side being incompatible qubit
measurements. Note that quantum violation of tripartite steering inequality (\ref{MIAEPR}) by $2 \times 2 \times 2$ systems implies quantum violation of the biseparability inequality (\ref{bosepnewequ}) by that $2 \times 2 \times 2$ systems. Because, for both of these two inequalities the upper bounds are the same. Hence the claim.
\end{proof} 

We have illustrated the above results with the Mermin family in Fig. \ref{MeFF}.

\section{Conclusion}\label{Cnc}
In this work, we have studied tripartite EPR steering of quantum correlations arising from two local measurements on each side in the two types of partially device-independent scenarios: $1$-sided device-independent scenario where one of the parties performs untrusted measurements while the other two parties perform trusted measurements and $2$-sided device-independent scenario where one of the parties performs trusted measurements while the other two parties perform untrusted measurements.

We have studied tripartite steering  and genuine tripartite steering in the $1$-sided device-independent framework in the following scenarios: one of the parties performs two dichotomic black-box measurements and the other two parties perform incompatible qubit measurements that
demonstrate Bell nonlocality \cite{CHS+69} in one of the types or perform incompatible measurements that demonstrate EPR steering without Bell nonlocality \cite{UFNL,jeba} in the other type.
In the context of these two scenarios, we have studied tripartite steering of two families of quantum correlations called 
Svetlichny family and Mermin family, respectively. We have shown that the ranges in which these families detect tripartite steering and genuine tripartite steering 
are the same. 

On the other hand, in the $2$-sided device-independent framework, the ranges in which the Svetlichny family and Mermin family detect tripartite steering and genuine tripartite steering are different. These studies reveal that tripartite steering in the $2$-sided device-independent scenario is weaker than tripartite steering in the $1$-sided device-independent scenario. That is
the Svetlichny family and Mermin family in the $2$-sided device-independent framework detect tripartite entanglement  for a larger region than that in the $1$-sided device-independent framework.  
Using biseparability inequality, it has been demonstrated that tripartite steering in the $2$-sided device-independent framework implies the presence of genuine tripartite entanglement of $2\times 2 \times 2$ quantum system, even if the correlation does not exhibit genuine nonlocality or genuine steering.

Similar to our tripartite $1$-sided device-independent scenario considered in Section \ref{sts}  where the trusted parties Alice and Bob perform  incompatible qubit measurements that demonstrate
Bell-CHSH inequality violation, in Ref. \cite{KZ96}, the authors considered a tripartite measurement scenario in which Alice and Bob perform incompatible qubit measurements that demonstrate maximal
Bell-CHSH inequality violation. In the latter scenario, the authors studied an interesting feature of genuinely tripartite entangled states called ``entangled entanglement" in which entanglement between measurement choices of Charlie and entanglement of the conditional states prepared on Alice and Bob's side by these measurement choices occurs. On the other hand, in our work, we have studied steerability between measurement choices of Charlie and entanglement of the conditional states prepared on Alice and Bob's side by these measurements on genuinely tripartite entangled states.

Note that in Ref. \cite{stm3} the definition of genuine tripartite steering was presented and it was experimentally demonstrated in \cite{stm4}. In their approach Alice, Bob and Charlie are all assumed to perform characterised measurements at some point (i. e., the trusted and untrusted parties are not fixed in their definition of genuine tripartite steering). On the other hand, trusted and untrusted parties are fixed (Alice, Bob are trusted and Charlie is untrusted in $1$SDI scenario;  Alice is trusted and Bob, Charlie are untrusted in $2$SDI scenario) in the definitions of tripartite steering and genuine tripartite steering presented in this paper which is an advancement in the context of the notion of tripartite steering. In the steering scenarios considered in Ref. \cite{cava} noisy GHZ state demonstrates genuine tripartite steering in $1$SDI scenario in a larger region compared to that in $2$SDI scenario. On the other hand, the two examples of quantum correlations presented in this study reveal that noisy GHZ state demonstrates tripartite steering (not genuine) in $2$SDI scenario in a larger region compared to that in $1$SDI scenario. One important point to be stressed here is that the procedures to detect genuine tripartite steering adopted in Refs. \cite{cava, stm6, stm7} are based on numerical calculations with the help of semidefinite program (SDP). But the advantage of our study is that the steering LHS-LHV model and the fully LHS-LHV model of the two families of correlations are derived analytically, not using SDP. The application of tripartite steering in the context of randomness certification has been studied in \cite{stm6} using SDP. It is worth to be studied in future what advantage one can gain in the context of randomness certification in tripartite steering scenario considered by us in the present study.

In Ref. \cite{jeba}, the author proposed two  inequalities for detecting genuine steering in the Svetlichny-type and Mermin-type one-sided device-independent scenarios. We have demonstrated that these inequalities do not detect genuine steering and they detect  tripartite steering of $2 \times d \times d$ systems in the $2$-sided device-independent framework. Further, the author argued that the violation of one of these inequalities imply genuine entanglement if one assumes only dimension of the trusted parties to be qubit
dimension. However, the present study demonstrates that the violation of these inequalities do not detect genuine entanglement in this context, on the other hand, the violation of those inequalities may imply genuine entanglement  in the scenario  where the dimensions of all three parties are assumed to be qubit dimension.

\section*{Acknowledgement} 
Authors are thankful to the anonymous referee for drawing their attention to Ref. \cite{KZ96}.
Authors are thankful to Prof. Guruprasad Kar and Dr. Nirman Ganguly for fruitful discussions. CJ is thankful to Prof. Paul Skrzypczyk for useful discussions
during the 657.WE-Heraeus Seminar ``Quantum Correlations in Space and Time''. 
DD acknowledges the financial support from University Grants Commission (UGC), Government of India. BB, CJ and DS acknowledge the financial support from project SR/S2/LOP-08/2013 of the Department of Science and Technology (DST), government of India. 

\appendix
\section{Proof for Lemma \ref{gscssf}}\label{PSLHVSF}
We consider the following classical simulation scenario to demonstrate in which range the Svetlichny family
 has a steering LHS-LHV model as in Eq. (\ref{NLHS})  and  a fully LHS-LHV model as in Eq. (\ref{ALHS}) in our $1$SDI scenario considered in Section \ref{sts}: 
\begin{scenario} \label{csss}
  Charlie generates his outcomes  by using classical variable $\lambda$ which he shares with Alice-Bob.
  Alice and Bob share a two-qubit system for each value of $\lambda$ and perform pair  of incompatible qubit measurements that demonstrate 
  Bell nonlocality of certain two-qubit states; for instance, the singlet state.
\end{scenario}

For $0 < V \le\frac{1}{\sqrt{2}}$, the Svetlichny family given by Eq.(\ref{SvF}) can be written as
\begin{equation}
\label{lhv-lhsSv}
P_{SvF}^{V} (abc|A_x B_y C_z)
= \sum_{\lambda=0}^{3} r_{\lambda} P(a b |A_x B_y, \rho^{\lambda}_{AB}) P_{\lambda} (c|C_z)
\end{equation}
where $r_0$ = $r_1$ = $r_2$ = $r_3$ = $\frac{1}{4}$, and \\
$P_{0} (c|C_z)$ = $P_D^{00}$, $P_{1} (c|C_z)$ = $P_D^{01}$, $P_{2} (c|C_z)$ = $P_D^{10}$, $P_{3} (c|C_z)$ = $P_D^{11}$,\\
here, \begin{equation}
P_D^{\alpha\beta}(c|C_z)=\left\{
\begin{array}{lr}
1, & c=\alpha z\oplus \beta\\
0 , & \text{otherwise}\\
\end{array}
\right. 
\end{equation}
Here, $\alpha, \beta \in \{0, 1\}$. The four bipartite distributions $P(a b |A_x B_y, \rho^{\lambda}_{AB})$ in Eq. (\ref{lhv-lhsSv})
are given as follows:
\begin{enumerate}
 \item For $\lambda=0$, it is given by,

 \begin{equation}
 P(a b |A_x B_y, \rho^{0}_{AB}) = \bordermatrix{
\frac{ab}{xy} & 00 & 01 & 10 & 11 \cr
00 & \frac{1+ \sqrt{2} V}{4} & \frac{1- \sqrt{2} V}{4} & \frac{1- \sqrt{2} V}{4} & \frac{1+ \sqrt{2} V}{4} \cr
01 & \frac{1}{4} & \frac{1}{4} & \frac{1}{4} & \frac{1}{4} \cr
10 & \frac{1}{4} & \frac{1}{4} & \frac{1}{4} & \frac{1}{4} \cr
11 & \frac{1- \sqrt{2} V}{4} & \frac{1+ \sqrt{2} V}{4} & \frac{1+ \sqrt{2} V}{4} & \frac{1- \sqrt{2} V}{4} },
\end{equation}
where each row and column corresponds to a fixed measurement $(xy)$ and a fixed outcome $(ab)$ respectively. Throughout the paper we will follow the same convention. Note that, each of the probability distributions must satisfy $0 \leq P(a b |A_x B_y, \rho^{0}_{AB}) \leq 1$, which implies that $0< V \leq \frac{1}{\sqrt{2}}$.

This  joint probability distribution at Alice and Bob's side can be reproduced by performing measurements of the observables corresponding to the operators $A_0=\sigma_x$, $A_1=\sigma_y$; and $B_0=\frac{\sigma_x-\sigma_y}{\sqrt{2}}$, $B_1=\frac{\sigma_x + \sigma_y}{\sqrt{2}}$  on the two-qubit state given by,
\begin{equation}
\label{states1}
| \psi_0 \rangle = \cos \theta |00 \rangle +  \frac{(1 - i) \sin \theta}{\sqrt{2}} |11 \rangle ,
\end{equation}  
with $ \sin 2 \theta = \sqrt{2} V$; $0 \leq \theta \leq \frac{\pi}{4}$. $|0\rangle$ and $|1\rangle$ are the eigenstates of $\sigma_z$ corresponding to the eigenvalues $+1$ and $-1$ respectively.

\item For $\lambda=1$, it is given by,

\begin{equation}
P(a b |A_x B_y, \rho^{1}_{AB}) = \begin{pmatrix}
\frac{1- \sqrt{2} V}{4} && \frac{1+ \sqrt{2} V}{4} && \frac{1+ \sqrt{2} V}{4} && \frac{1- \sqrt{2} V}{4}\\
\frac{1}{4} && \frac{1}{4} && \frac{1}{4} && \frac{1}{4} \\
\frac{1}{4} && \frac{1}{4} && \frac{1}{4} && \frac{1}{4} \\
\frac{1+ \sqrt{2} V}{4} && \frac{1- \sqrt{2} V}{4} && \frac{1- \sqrt{2} V}{4} && \frac{1+ \sqrt{2} V}{4}\\
\end{pmatrix}.
\end{equation}
Note that, each of the probability distributions must satisfy $0 \leq P(a b |A_x B_y, \rho^{1}_{AB}) \leq 1$, which implies that $0< V \leq \frac{1}{\sqrt{2}}$.\\

This joint probability distribution at Alice and Bob's side can be reproduced by performing measurements of the observables corresponding to the operators $A_0=\sigma_x$, $A_1=\sigma_y$; and $B_0=\frac{\sigma_x-\sigma_y}{\sqrt{2}}$, $B_1=\frac{\sigma_x + \sigma_y}{\sqrt{2}}$ on the two-qubit state given by,
\begin{equation}
\label{states2}
| \psi_1 \rangle = \cos \theta |00 \rangle - \frac{(1 - i) \sin \theta}{\sqrt{2}} |11 \rangle ,
\end{equation} 
with $ \sin 2 \theta = \sqrt{2} V$; $0 \leq \theta \leq \frac{\pi}{4}$.\\

\item For $\lambda=2$, it is given by,

\begin{equation}
P(a b |A_x B_y, \rho^{2}_{AB}) = \begin{pmatrix}
\frac{1}{4} && \frac{1}{4} && \frac{1}{4} && \frac{1}{4} \\
\frac{1+ \sqrt{2} V}{4} && \frac{1- \sqrt{2} V}{4} && \frac{1- \sqrt{2} V}{4} && \frac{1+ \sqrt{2} V}{4}\\
\frac{1+ \sqrt{2} V}{4} && \frac{1- \sqrt{2} V}{4} && \frac{1- \sqrt{2} V}{4} && \frac{1+ \sqrt{2} V}{4}\\
\frac{1}{4} && \frac{1}{4} && \frac{1}{4} && \frac{1}{4} \\
\end{pmatrix} .
\end{equation}
Note that, each of the probability distributions must satisfy $0 \leq P(a b |A_x B_y, \rho^{2}_{AB}) \leq 1$, which implies that $0< V \leq \frac{1}{\sqrt{2}}$.\\

This joint probability distribution at Alice and Bob's side can be reproduced by performing measurements of the observables corresponding to the operators $A_0=\sigma_x$, $A_1=\sigma_y$; and $B_0=\frac{\sigma_x-\sigma_y}{\sqrt{2}}$, $B_1=\frac{\sigma_x + \sigma_y}{\sqrt{2}}$ on the two-qubit  state given by,
\begin{equation}
\label{states3}
| \psi_2 \rangle = \cos \theta |00 \rangle +  \frac{(1 + i) \sin \theta}{\sqrt{2}} |11 \rangle ,
\end{equation} 
with $ \sin 2 \theta = \sqrt{2} V$; $0 \leq \theta \leq \frac{\pi}{4}$.\\

\item For $\lambda=3$, it is given by,

\begin{equation}
P(a b |A_x B_y, \rho^{3}_{AB}) = \begin{pmatrix}
\frac{1}{4} && \frac{1}{4} && \frac{1}{4} && \frac{1}{4} \\
\frac{1- \sqrt{2} V}{4} && \frac{1+ \sqrt{2} V}{4} && \frac{1+ \sqrt{2} V}{4} && \frac{1- \sqrt{2} V}{4}\\
\frac{1- \sqrt{2} V}{4} && \frac{1+ \sqrt{2} V}{4} && \frac{1+ \sqrt{2} V}{4} && \frac{1- \sqrt{2} V}{4}\\
\frac{1}{4} && \frac{1}{4} && \frac{1}{4} && \frac{1}{4} \\
\end{pmatrix}.
\end{equation}
Note that, each of the probability distributions must satisfy $0 \leq P(a b |A_x B_y, \rho^{3}_{AB}) \leq 1$, which implies that $0< V \leq \frac{1}{\sqrt{2}}$.\\

This joint probability distribution at Alice and Bob's side can be reproduced by performing measurements of the observables corresponding to the operators $A_0=\sigma_x$, $A_1=\sigma_y$; and $B_0=\frac{\sigma_x-\sigma_y}{\sqrt{2}}$, $B_1=\frac{\sigma_x + \sigma_y}{\sqrt{2}}$ on the two-qubit  state given by,
\begin{equation}
\label{states4}
| \psi_3 \rangle = \cos \theta |00 \rangle -  \frac{(1 + i) \sin \theta}{\sqrt{2}} |11 \rangle ,
\end{equation} 
with $ \sin 2 \theta = \sqrt{2} V$; $0 \leq \theta \leq \frac{\pi}{4}$.

\end{enumerate}

Here it can be easily checked that the aforementioned observables corresponding to the operators $A_0=\sigma_x$, $A_1=\sigma_y$; and $B_0=\frac{\sigma_x-\sigma_y}{\sqrt{2}}$, $B_1=\frac{\sigma_x + \sigma_y}{\sqrt{2}}$ used to reproduce the joint probability distributions at Alice and Bob's side can demonstrate nonlocality of the singlet state given by, $| \psi^{-} \rangle = \frac{1}{\sqrt{2}} ( | 01 \rangle - | 10 \rangle )$. Hence, the Svetlichny family given by Eq.(\ref{SvF}) has a steering LHS-LHV model as in Eq. (\ref{NLHS})  in the range $0 < V \le\frac{1}{\sqrt{2}}$ in Scenario \ref{csss}.

In the steering LHS-LHV model given for the Svetlichny family as in Eq.(\ref{lhv-lhsSv}), 
 the bipartite distributions   $P(a b |A_x B_y, \rho^{\lambda}_{AB})$ belong 
 to the BB84 family up to local reversible operations (LRO)
 \footnote{LRO is designed  \cite{lro} as follows: 
 Alice may relabel her inputs: $x \rightarrow x \oplus 1$, and she may relabel her outputs (conditionally on the input) : $a \rightarrow a \oplus \alpha x \oplus \beta$ ($\alpha, \beta \in \{0, 1\}$); Bob can perform similar operations.},
\begin{equation}
\label{bb84}
P_{BB84}(ab|A_xB_y) = \frac{1 + (-1)^{a \oplus b \oplus x.y} \delta_{x,y} W }{4}
\end{equation}
where $W=\sqrt{2}V$ is a real number such that $0<W\le 1$. In Ref. \cite{Koon},
it has been shown that the BB84 family certifies two-qubit entanglement iff $W>\frac{1}{2}$.
This implies that for $W\le \frac{1}{2}$, it can be reproduced by a two-qubit separable
state. Therefore, 
the bipartite distributions   $P(a b |A_x B_y, \rho^{\lambda}_{AB})$
in Eq. (\ref{lhv-lhsSv}) has a LHS-LHS decomposition for $V\le \frac{1}{2 \sqrt{2}}$.
This implies that the Svetlichny family can be reproduced by a fully LHS-LHV model,
\be
P_{SvF}^{V} (abc|A_x B_ y C_z)=\sum_\lambda q_\lambda P(a|A_x,\rho_A^\lambda)P(b|B_y,\rho_B^\lambda)P_\lambda(c|C_z), \label{ALHS333}
\ee
for $V \leq \frac{1}{2\sqrt{2}}$ in Scenario \ref{csss}.
Here, $P(a|A_x,\rho_A^\lambda)$ 
and $P(b|B_y,\rho_B^\lambda)$ are the distributions arising from the local hidden states 
$\rho_A^\lambda$ and $\rho_B^\lambda$ which are in  $\mathbb{C}^{2}$, respectively.

In Ref. \cite{UFNL}, it has been shown that violation of the following inequality (Eq. (22) in \cite{UFNL} with $N$ (Number of parties) $=3$ and $T$ (Number of trusted parties) $=2$):
\ba \label{eqappn1}
\braket{S}_{2\times 2 \times ?} \overset{\mathrm{LHS}}{\le} 2,
\ea
detects non-existence of fully LHS-LHV model in Scenario \ref{csss}. Here, $S$ is the Svetlichny operator given in the Svetlichny inequality (\ref{SI1}), $2\times 2 \times ?$ indicates that Alice and Bob perform qubit measurements
while Charlie performs black-box measurements. Note that the Svetlichny family  
violates the above inequality for $V>\frac{1}{2\sqrt{2}}$. Thus, the Svetlichny family does not have fully LHS-LHV model in the region $V>\frac{1}{2\sqrt{2}}$ in Scenario \ref{csss}. Hence, we can conclude that the Svetlichny family has fully LHS-LHV model iff $0 < V \leq \frac{1}{2\sqrt{2}}$ in Scenario \ref{csss}.

\section{Derivation of the Svetlichny biseparability inequality} \label{bisepapp}
Here we derive a biseparability inequality 
that detect 
genuine entanglement of three-qubit systems by using the Svetlichny operator in the scenario
where each party performs incompatible  qubit measurements. In this scenario, the tripartite correlations that can be reproduced
by a biseparable three-qubit state has  the following nonseparable LHS-LHS (NSLHS) model:
\begin{align}
P(abc|A_xB_yC_z)\!&=\!\sum_\lambda p_\lambda P(a|A_x,\rho^\lambda_A)P(bc|B_yC_z,\rho^{\lambda}_{BC})\nonumber \\
&+\!\sum_\lambda q_\lambda P(ac|A_xC_z,\rho^\lambda_{AC})P(b|B_y, \rho^\lambda_B)\!\nonumber \\
&+\!\sum_\lambda r_\lambda P(ab|A_xB_y,\rho^\lambda_{AB})P(c|C_z,\rho^\lambda_C), \label{NSLHS}
\end{align}
with $\sum_\lambda p_\lambda+\sum_\lambda q_\lambda+\sum_\lambda r_\lambda=1$. Here,
$P(a|A_x,\rho^\lambda_A)$, $P(b|B_y, \rho^\lambda_B)$ and $P(c|C_z, \rho^\lambda_C)$ are 
the distributions which can be reproduced
by the qubit states $\rho^\lambda_A$, $\rho^\lambda_B$ and $\rho^\lambda_C$, respectively, 
and $P_\lambda(bc|B_yC_z,\rho^{\lambda}_{BC})$, $P_\lambda(ac|A_xC_z,\rho^\lambda_{AC})$ and
$P_\lambda(ab|A_xB_y,\rho^\lambda_{AB})$ can be reproduced by the $2 \times 2$ states $\rho^{\lambda}_{BC}$, 
$\rho^\lambda_{AC}$ and $\rho^\lambda_{AB}$, respectively.
Note that in the model given by Eq. (\ref{NSLHS}), the bipartite distributions  at each 
$\lambda$ level may have nonseparability.

The Svetlichny operator can be rewritten  as follows:
\be
S=CHSH_{AB}C_1+CHSH'_{AB}C_0. \label{SIO}
\ee
 Here, $CHSH_{AB} = A_0B_0+A_0B_1+A_1B_0-A_1B_1$ is the canonical CHSH (Clauser-Horne-Shimony-Holt) operator \cite{CHS+69}  and $CHSH'_{AB}=-A_0B_0+A_0B_1+A_1B_0+A_1B_1$ is one of its equivalents.
 Note that the expectation value of the Svetlichny operator for the correlation which has 
the nonseparable LHS-LHS model as given in Eq. (\ref{NSLHS}) have the following form:
\begin{align}
&\sum_\lambda p_\lambda \braket{A_1}_{\rho^\lambda_A}\braket{CHSH_{BC}}_{\rho^\lambda_{BC}}+
 \sum_\lambda p_\lambda \braket{A_0}_{\rho^\lambda_A}\braket{CHSH'_{BC}}_{\rho^\lambda_{BC}} \nonumber \\
&+\sum_\lambda q_\lambda \braket{CHSH_{AC}}_{\rho^\lambda_{AC}}\braket{B_1}_{\rho^\lambda_B}+
\sum_\lambda q_\lambda\braket{CHSH'_{AC}}_{\rho^\lambda_{AC}}\braket{B_0}_{\rho^\lambda_B} \nonumber\\
&+\sum_\lambda r_\lambda \braket{CHSH_{AB}}_{\rho^\lambda_{AB}}\braket{C_1}_{\rho^\lambda_C}+
\sum_\lambda r_\lambda\braket{CHSH'_{AB}}_{\rho^\lambda_{AB}}\braket{C_0}_{\rho^\lambda_C}. \label{SObi}
\end{align}
Let us now argue that the above quantity is upper bounded by $2\sqrt{2}$.
Consider the first line of the decomposition given in Eq. (\ref{SObi}). 
Suppose Bob and Charlie's correlation at each $\lambda$ level of this line detects nonseparability.
Then $\pm\braket{CHSH_{BC}}_{\rho^\lambda_{BC}}\pm
\braket{CHSH'_{BC}}_{\rho^\lambda_{BC}} \le 2\sqrt{2}$.
Suppose Bob and Charlie's correlation at each $\lambda$ level has a LHS-LHS model.
Then also $\pm\braket{CHSH_{BC}}_{\rho^\lambda_{BC}}\pm
\braket{CHSH'_{BC}}_{\rho^\lambda_{BC}} \le 2\sqrt{2}$. In a similar way, considering the second line of the decomposition given in Eq. (\ref{SObi}), one can show that $\pm\braket{CHSH_{AC}}_{\rho^\lambda_{AC}}\pm
\braket{CHSH'_{AC}}_{\rho^\lambda_{AC}} \le 2\sqrt{2}$; and considering the third line of the decomposition given in Eq. (\ref{SObi}), one can show that $\pm\braket{CHSH_{AB}}_{\rho^\lambda_{AB}}\pm
\braket{CHSH'_{AB}}_{\rho^\lambda_{AB}} \le 2\sqrt{2}$. Therefore, any convex combination of the three above mentioned expression should be upper bounded by $2\sqrt{2}$. Hence,
we can conclude that in the Scenario where each party performs incompatible qubit measurements, the Svetlichny operator is upper bounded
by $2\sqrt{2}$ if the correlation has a nonseparable LHS-LHS model (\ref{NSLHS}).
Hence the following   inequality: 
\be
\braket{S}_{2\times 2 \times 2} \overset{\mathrm{Bi-sep}}{\le} 2\sqrt{2}, 
\ee
serves as the biseparability inequality whose violation detects genuine tripartite entanglement of $2 \times 2 \times 2$ systems.
Here, $\braket{S}_{2\times 2 \times 2}$ denotes the Svetlichny operator with the measurement observables on each side being incompatible qubit
measurements.

\section{Proof for Lemma \ref{gscsmf}}\label{PSLHVMF}
We consider the following classical simulation scenario to demonstrate in which range the Mermin family
 has a steering LHS-LHV model as in Eq. (\ref{NLHS})  and  a fully LHS-LHV model as in Eq. (\ref{ALHS}) in the $1$SDI scenario considered in Section \ref{mts}: 
\begin{scenario} \label{csms}
  Charlie generates his outcomes  by using classical variable $\lambda$ which he shares with Alice-Bob.
  Alice and Bob share a two-qubit system for each $\lambda$ and perform pair  of incompatible qubit measurements that demonstrate
  EPR steering without Bell nonlocality of certain two-qubit states \cite{UFNL,jeba}; for instance, the singlet state.
\end{scenario}

Following the steering LHV-LHS model of the Mermin family mentioned in Ref. \cite{DJB+17},
we can write down the following steering LHS-LHV model of the Mermin family in the range $0 < V \le\frac{1}{\sqrt{2}}$,
\begin{equation}
\label{lhvlhs}
P_{MF}^{V} (abc|A_x B_ y C_z) = \sum_{\lambda=0}^{3} r_{\lambda} P(a b |A_x B_y, \rho^{\lambda}_{AB}) P_{\lambda} (c|C_z),
\end{equation}
 as it is invariant under the permutations of the parties.
Here, $r_0$ = $r_1$ = $r_2$ = $r_3$ = $\frac{1}{4}$, and \\
$P_{0} (c|C_z)$ = $P_D^{00}$, $P_{1} (c|C_z)$ = $P_D^{01}$, $P_{2} (c|C_z)$ = $P_D^{10}$, $P_{3} (c|C_z)$ = $P_D^{11}$.

The bipartite distributions in the model (\ref{lhvlhs}) are given as follows:
\begin{enumerate}
 \item For $\lambda=0$, it is given by
\begin{equation}
P(a b |A_x B_y, \rho^{0}_{AB}) = \bordermatrix{
\frac{ab}{xy} & 00 & 01 & 10 & 11 \cr
00 & \frac{1+V}{4} & \frac{1-V}{4} & \frac{1-V}{4} & \frac{1+V}{4} \cr
01 & \frac{1+V}{4} & \frac{1-V}{4} & \frac{1-V}{4} & \frac{1+V}{4} \cr
10 & \frac{1+V}{4} & \frac{1-V}{4} & \frac{1-V}{4} & \frac{1+V}{4} \cr
11 & \frac{1-V}{4} & \frac{1+V}{4} & \frac{1+V}{4} & \frac{1-V}{4} },
\end{equation}
where each row and column corresponds to a fixed measurement $(xy)$ and a fixed outcome $(ab)$ respectively. This joint probability can be reproduced by performing the projective qubit measurements of the observables corresponding to the operators  $A_0=\sigma_x$, $A_1=\sigma_y$; and $B_0=\sigma_x$, $B_1=\sigma_y$ on the two-qubit state given by,
\begin{equation}
\label{state1}
| \psi_0 \rangle = \cos \theta |00 \rangle +  \frac{(1 + i) \sin \theta}{\sqrt{2}} |11 \rangle ,
\end{equation}
where, $0 \leq \theta \leq \frac{\pi}{4}$ with $\sin 2 \theta = \sqrt{2} V$; $|0\rangle$ and $|1\rangle$ are the eigenstates of $\sigma_z$ corresponding to the eigenvalues $+1$ and $-1$ respectively.

\item For $\lambda=1$, it is given by
\begin{center}
$ P(a b |A_x B_y, \rho^{1}_{AB}) = \begin{pmatrix}
\frac{1-V}{4} && \frac{1+V}{4} && \frac{1+V}{4} && \frac{1-V}{4}\\
\frac{1-V}{4} && \frac{1+V}{4} && \frac{1+V}{4} && \frac{1-V}{4}\\
\frac{1-V}{4} && \frac{1+V}{4} && \frac{1+V}{4} && \frac{1-V}{4}\\
\frac{1+V}{4} && \frac{1-V}{4} && \frac{1-V}{4} && \frac{1+V}{4}\\
\end{pmatrix}, $ 
\end{center} 
which can be reproduced by performing the projective qubit measurements of the observables corresponding to the operators $A_0=\sigma_x$, $A_1=\sigma_y$; and $B_0=\sigma_x$, $B_1=\sigma_y$ on the two-qubit state given by,
\begin{equation}
\label{state2}
| \psi_1 \rangle =  \cos \theta |00 \rangle -  \frac{(1 + i) \sin \theta}{\sqrt{2}} |11 \rangle ,
\end{equation}
where, $0 \leq \theta \leq \frac{\pi}{4}$ with $\sin 2 \theta = \sqrt{2} V$.

\item For $\lambda=2$, it is given by
\begin{center}
$ P(a b |A_x B_y, \rho^{2}_{AB}) = \begin{pmatrix}
\frac{1-V}{4} && \frac{1+V}{4} && \frac{1+V}{4} && \frac{1-V}{4}\\
\frac{1+V}{4} && \frac{1-V}{4} && \frac{1-V}{4} && \frac{1+V}{4}\\
\frac{1+V}{4} && \frac{1-V}{4} && \frac{1-V}{4} && \frac{1+V}{4}\\
\frac{1+V}{4} && \frac{1-V}{4} && \frac{1-V}{4} && \frac{1+V}{4}\\
\end{pmatrix} ,$ 
\end{center}
which  can be reproduced by performing the projective qubit measurements of the observables corresponding to the operators $A_0=\sigma_x$, $A_1=\sigma_y$; and $B_0=\sigma_x$, $B_1=\sigma_y$ on the two-qubit state given by
\begin{equation}
\label{state3}
| \psi_2 \rangle =  \cos \theta |00 \rangle -  \frac{(1 - i) \sin \theta}{\sqrt{2}} |11 \rangle ,
\end{equation}
where, $0 \leq \theta \leq \frac{\pi}{4}$ with $\sin 2 \theta = \sqrt{2} V$.

\item For $\lambda=3$, it is given by
\begin{center}
$ P(a b |A_x B_y, \rho^{3}_{AB}) = \begin{pmatrix}
\frac{1+V}{4} && \frac{1-V}{4} && \frac{1-V}{4} && \frac{1+V}{4}\\
\frac{1-V}{4} && \frac{1+V}{4} && \frac{1+V}{4} && \frac{1-V}{4}\\
\frac{1-V}{4} && \frac{1+V}{4} && \frac{1+V}{4} && \frac{1-V}{4}\\
\frac{1-V}{4} && \frac{1+V}{4} && \frac{1+V}{4} && \frac{1-V}{4}\\
\end{pmatrix} ,$ 
\end{center}
which can be reproduced by performing the projective qubit measurements of the observables corresponding to the operators $A_0=\sigma_x$, $A_1=\sigma_y$; and $B_0=\sigma_x$, $B_1=\sigma_y$ on the two-qubit state given by
\begin{equation}
\label{state4}
| \psi_3 \rangle = \cos \theta |00 \rangle +  \frac{(1 - i) \sin \theta}{\sqrt{2}} |11 \rangle ,
\end{equation}
where, $0 \leq \theta \leq \frac{\pi}{4}$ with $\sin 2 \theta = \sqrt{2} V$.
\end{enumerate}
Note that $|\sin 2 \theta | \leq 1$ (as $0 \leq \theta \leq \frac{\pi}{4}$), which implies that $V \leq \frac{1}{\sqrt{2}}$.  It can be easily checked that the aforementioned observables corresponding to the operators $A_0=\sigma_x$, $A_1=\sigma_y$; and $B_0=\sigma_x$, $B_1=\sigma_y$ used to reproduce the joint probability distributions at Alice and Bob's side can demonstrate  EPR steering without Bell nonlocality of the singlet state given by, $| \psi^{-} \rangle = \frac{1}{\sqrt{2}} ( | 01 \rangle - | 10 \rangle )$. Hence, the Mermin family given by Eq.(\ref{mfameq}) has a steering LHS-LHV model as in Eq.(\ref{NLHS})  in the range $0 < V \le\frac{1}{\sqrt{2}}$ in Scenario \ref{csms}.\\

In the steering LHS-LHV model given for the Mermim family as in Eq.(\ref{lhvlhs}), 
 the bipartite distributions   $P(a b |A_x B_y, \rho^{\lambda}_{AB})$ belong 
 to the CHSH family up to local reversible operations \cite{lro},
\begin{equation}
\label{chsh}
P_{CHSH}(ab|A_xB_y)       =       \frac{2+(-1)^{a\oplus       b\oplus
    xy}\sqrt{2}W}{8} ,
\end{equation}
where $W=\sqrt{2}V$ is a real number such that $0<W\le 1$ and $0 < V \leq \frac{1}{\sqrt{2}}$.
In Ref. \cite{Koon},
it has been that the CHSH family certifies two-qubit entanglement iff $W>\frac{1}{2}$.
This implies that for $W\le \frac{1}{2}$, it can be reproduced by a two-qubit separable
state. Therefore, the bipartite distributions   $P(a b |A_x B_y, \rho^{\lambda}_{AB})$
in Eq. (\ref{lhvlhs}) has a LHS-LHS decomposition for $V\le\frac{1}{2\sqrt{2}}$.
This implies that the Mermin family can be reproduced by a fully LHS-LHV model,
\be
P_{MF}^{V} (abc|A_x B_ y C_z)=\sum_\lambda q_\lambda P(a|A_x,\rho_A^\lambda)P(b|B_y,\rho_B^\lambda)P_\lambda(c|C_z), \label{ALHS3}
\ee
for $V \leq \frac{1}{2\sqrt{2}}$ in Scenario \ref{csms}.
Here, $P(a|A_x,\rho_A^\lambda)$ 
and $P(b|B_y,\rho_B^\lambda)$ are the distributions arising from the local hidden states 
$\rho_A^\lambda$ and $\rho_B^\lambda$ which are in  $\mathbb{C}^{2}$, respectively.

In Ref. \cite{UFNL}, it has been shown that violation of the following inequality (Eq. ($21$) in \cite{UFNL} with $N=3$ and $T=2$),
\ba \label{eqappn2}
\braket{M}_{2\times 2 \times ?} \overset{\mathrm{LHS}}{\le} \sqrt{2},
\ea
detects non-existence of fully LHS-LHV model in Scenario \ref{csms}. Here $M$ is the Mermin operator given in the Mermin inequality (\ref{MI0}),
$2\times 2 \times ?$ indicates that Alice and Bob perform qubit measurements
while Charlie performs black-box measurements. Note that the Mermin family  
violates the above inequality for $V>\frac{1}{2\sqrt{2}}$. Thus, the Mermin family does not have fully LHS-LHV model in the region $V>\frac{1}{2\sqrt{2}}$ in Scenario \ref{csms}. Hence, we can conclude that the Mermin family has fully LHS-LHV model iff $0 < V \leq \frac{1}{2\sqrt{2}}$ in Scenario \ref{csms}.


\begin{thebibliography}{99}
%\bibitem{bell} J. S. Bell, Physics 1, 195 (1964).
%\bibitem{CHSH} J. F. Clauser, M. A. Horne, A. Shimony, and R A Holt, Phys. Rev. Lett. {\bf 23}, 880 (1969).
\bibitem{guhne} O. Guhne, and G. Toth, \emph{Entanglement detection}, \href{https://www.sciencedirect.com/science/article/pii/S0370157309000623}{Physics Reports \textbf{474}, 1 (2009).}
\bibitem{Sor} A. S. Sorensen, and K. Molmer, \emph{Entanglement and Extreme Spin Squeezing}, \href{https://journals.aps.org/prl/abstract/10.1103/PhysRevLett.86.4431}{Phys. Rev. Lett. \textbf{86}, 4431 (2001).}
\bibitem{Hyl} P. Hyllus, W. Laskowski, R. Krischek, C. Schwemmer, W. Wieczorek, H. Weinfurter, L. Pezze, and A. Smerzi, \emph{Fisher information and multiparticle entanglement}, \href{https://journals.aps.org/pra/abstract/10.1103/PhysRevA.85.022321}{Phys. Rev. A \textbf{85}, 022321 (2012).}
\bibitem{toth} G. Toth, \emph{Multipartite entanglement and high-precision metrology}, \href{https://journals.aps.org/pra/abstract/10.1103/PhysRevA.85.022322}{Phys. Rev. A \textbf{85}, 022322 (2012).}
\bibitem{bell} J. S. Bell, \emph{On the Einstein-Podolsky-Rosen paradox}, \href{https://cds.cern.ch/record/111654/files/vol1p195-200_001.pdf}{Physics \textbf{1}, 195 (1965)}.
\bibitem{SI} G. Svetlichny, \emph{Distinguishing three-body from two-body nonseparability by a Bell-type inequality}, \href{https://journals.aps.org/prd/abstract/10.1103/PhysRevD.35.3066}{Phys. Rev. D \textbf{35}, 3066 (1987).}
\bibitem{Nag} K. Nagata, M. Koashi, and N. Imoto, \emph{Configuration of Separability and Tests for Multipartite Entanglement in Bell-Type Experiments}, \href{https://journals.aps.org/prl/abstract/10.1103/PhysRevLett.89.260401}{Phys. Rev. Lett. \textbf{89}, 260401 (2002).}
\bibitem{mb1} D. Collins, N. Gisin, S. Popescu, D. Roberts, and V. Scarani, \emph{Bell-Type Inequalities to Detect True $n$-Body Nonseparability}, \href{https://journals.aps.org/prl/abstract/10.1103/PhysRevLett.88.170405}{Phys. Rev. Lett. \textbf{88}, 170405 (2002).}
\bibitem{mb2} M. Seevinck, and G. Svetlichny, \emph{Bell-Type Inequalities for Partial Separability in $N$-Particle Systems and Quantum Mechanical Violations}, \href{https://journals.aps.org/prl/abstract/10.1103/PhysRevLett.89.060401}{Phys. Rev. Lett. \textbf{89}, 060401(2002).}
\bibitem{B} J.-D. Bancal, J. Barrett, N. Gisin, and S. Pironio, \emph{Definitions of multipartite nonlocality}, \href{https://journals.aps.org/pra/abstract/10.1103/PhysRevA.88.014102}{Phys. Rev. A \textbf{88}, 014102 (2013).}
\bibitem{EPR} A. Einstein, B. Podolsky, and N. Rosen, \emph{Can Quantum-Mechanical Description of Physical Reality Be Considered Complete?}, \href{https://journals.aps.org/pr/abstract/10.1103/PhysRev.47.777}{Phys. Rev. \textbf{47}, 777 (1935).}
\bibitem{mb3} D. Saha, and M. Pawlowski, \emph{Structure of quantum and broadcasting nonlocal correlations}, \href{https://journals.aps.org/pra/abstract/10.1103/PhysRevA.92.062129}{Phys. Rev. A \textbf{92}, 062129 (2015).}
\bibitem{scro} E. Schrodinger, Proc. Cambridge Philos. Soc. {\bf 31}, 553 (1935); {\bf 32}, 446 (1936).
\bibitem{steer} H. M. Wiseman, S. J. Jones, and A. C. Doherty, \emph{Steering, Entanglement, Nonlocality, and the Einstein-Podolsky-Rosen Paradox}, \href{https://journals.aps.org/prl/abstract/10.1103/PhysRevLett.98.140402}{Phys. Rev. Lett. \textbf{98}, 140402 (2007).}
\bibitem{ZHC16} H. Zhu, M. Hayashi, and L. Chen,  \emph{Universal Steering Criteria}, \href{https://journals.aps.org/prl/abstract/10.1103/PhysRevLett.116.070403}{Phys. Rev. Lett. \textbf{116}, 070403 (2016).}
\bibitem{CHS+69} J. F. Clauser, M. A. Horne, A. Shimony, and R. A. Holt, \emph{Proposed Experiment to Test Local Hidden-Variable Theories}, \href{https://journals.aps.org/prl/abstract/10.1103/PhysRevLett.23.880}{Phys. Rev. Lett. {\bf 23}, 880 (1969).}
\bibitem{UFNL} E. G. Cavalcanti, Q. Y. He, M. D. Reid, and H. M. Wiseman, \emph{Unified criteria for multipartite quantum nonlocality}, \href{https://journals.aps.org/pra/abstract/10.1103/PhysRevA.84.032115}{Phys. Rev. A \textbf{84}, 032115 (2011).}
\bibitem{stm2} Q. Y. He, P. D. Drummond, and M. D. Reid, \emph{Entanglement, EPR steering, and Bell-nonlocality criteria for multipartite higher-spin systems}, \href{https://journals.aps.org/pra/abstract/10.1103/PhysRevA.83.032120}{Phys. Rev. A \textbf{83}, 032120 (2011).}
\bibitem{stm3} Q. Y. He, and M. D. Reid, \emph{Genuine Multipartite Einstein-Podolsky-Rosen Steering}, \href{https://journals.aps.org/prl/abstract/10.1103/PhysRevLett.111.250403}{Phys. Rev. Lett. \textbf{111}, 250403 (2013).}
\bibitem{stm4} C.-M. Li, K. Chen, Y.-N. Chen, Q. Zhang, Y.-A. Chen, and J.-W. Pan, \emph{Genuine High-Order Einstein-Podolsky-Rosen Steering}, \href{https://journals.aps.org/prl/abstract/10.1103/PhysRevLett.115.010402}{Phys. Rev. Lett. \textbf{115}, 010402 (2015).}
\bibitem{cava} D. Cavalcanti, P. Skrzypczyk, G. H. Aguilar, R. V. Nery, P. S. Ribeiro, and S. P. Walborn, \emph{Detection of entanglement in asymmetric quantum networks and multipartite quantum steering}, \href{https://www.nature.com/articles/ncomms8941}{Nat. Commun. \textbf{6}, 7941 (2015).}
\bibitem{stm6} A. Mattar, P. Skrzypczyk, G. H. Aguilar, R. V. Nery, P. H. Souto Ribeiro, S. P. Walborn and D. Cavalcanti, \emph{Experimental multipartite entanglement and randomness certification of the W state in the quantum steering scenario}, \href{http://iopscience.iop.org/article/10.1088/2058-9565/aa629b/meta}{Quantum Sci. Technol. \textbf{2}, 015011 (2017).}
\bibitem{stm7} D. Cavalcanti, and P. Skrzypczyk, \emph{Quantum steering: a review with focus on semidefinite programming}, \href{http://iopscience.iop.org/article/10.1088/1361-6633/80/2/024001/meta}{Rep. Prog. Phys. \textbf{80}, 024001 (2017).}
\bibitem{jeba} C. Jebaratnam, \emph{Detecting genuine multipartite entanglement in steering scenarios}, \href{https://journals.aps.org/pra/abstract/10.1103/PhysRevA.93.052311}{Phys. Rev. A \textbf{93}, 052311 (2016).}
\bibitem{mermin} N. D. Mermin, \emph{Extreme quantum entanglement in a superposition of macroscopically distinct states}, \href{https://journals.aps.org/prl/abstract/10.1103/PhysRevLett.65.1838}{Phys. Rev. Lett. \textbf{65}, 1838 (1990).}
%\bibitem{dw} J. M. Donohue, and E. Wolfe, Phys. Rev. A \textbf{92}, 062120 (2015).
\bibitem{GHZ} D. M. Greenberger, M. A. Horne, and A. Zeilinger, \textit{Bell’s Theorem, Quantum Theory, and Conceptions of the Universe}, (1989).
\bibitem{mermin2} N. D. Mermin, \emph{Simple unified form for the major no-hidden-variables theorems}, \href{https://journals.aps.org/prl/abstract/10.1103/PhysRevLett.65.3373}{Phys. Rev. Lett. \textbf{65}, 3373 (1990).}
\bibitem{DQKD} A. Acin, N. Gisin, and L. Masanes, \emph{From Bell’s Theorem to Secure Quantum Key Distribution}, \href{https://journals.aps.org/prl/abstract/10.1103/PhysRevLett.97.120405}{Phys. Rev. Lett. \textbf{97}, 120405 (2006).}
\bibitem{Guhne} O. Guhne, and M. Seevinck, \emph{Separability criteria for genuine multiparticle entanglement}, \href{http://iopscience.iop.org/article/10.1088/1367-2630/12/5/053002/meta}{New J. Phys. \textbf{12}, 053002 (2010).}
%\bibitem{CFF+15} E. G. Cavalcanti, C. J. Foster, M. Fuwa, and H. M. Wiseman, J. Opt. Soc. Am. B %\textbf{32}, A74 (2015).
%\bibitem{FDB+16} D. Frustaglia, J. P. Baltanas, M. C. Velazquez-Ahumada, A. Fernandez-Prieto, A. %Lujambio, V. Losada, M. J. Freire, and A. Cabello, Phys. Rev. Lett. \textbf{116}, 250404 (2016).
\bibitem{Bancaletal} J.-D. Bancal, N. Brunner, N. Gisin, and Y.-C. Liang, \emph{Detecting Genuine Multipartite Quantum Nonlocality: A Simple Approach and Generalization to Arbitrary Dimensions}, \href{https://journals.aps.org/prl/abstract/10.1103/PhysRevLett.106.020405}{Phys. Rev. Lett. \textbf{106}, 020405 (2011).}
\bibitem{lro} J. Barrett, N. Linden, S. Massar, S. Pironio, S. Popescu, and D. Roberts, \emph{Nonlocal correlations as an information-theoretic resource}, \href{https://journals.aps.org/pra/abstract/10.1103/PhysRevA.71.022101}{Phys. Rev. A \textbf{71}, 022101 (2005).}
\bibitem{Koon} K. T. Goh, J-D. Bancal, and V. Scarani, \emph{Measurement-device-independent quantification of entanglement for given Hilbert space dimension}, \href{http://iopscience.iop.org/article/10.1088/1367-2630/18/4/045022}{New J. Phys. \textbf{18}, 045022 (2016).}
\bibitem{JGA17} C. Jebaratnam, D. Das, S. Goswami, R. Srikanth, A. S. Majumdar, \emph{Operational nonclassicality of local multipartite correlations in the limited-dimensional simulation scenario}, \href{http://iopscience.iop.org/article/10.1088/1751-8121/aad1fd/meta}{J. Phys. A: Math. Theor. in press https://doi.org/10.1088/1751-8121/aad1fd (2018).}
\bibitem{SU01}M. Seevinck and J. Uffink,  \emph{Sufficient conditions for three-particle entanglement and their tests in recent experiments}, \href{https://journals.aps.org/pra/abstract/10.1103/PhysRevA.65.012107}{Phys. Rev. A 65, 012107 (2001).}
\bibitem{DJB+17} D. Das, C. Jebaratnam, B. Bhattacharya, A. Mukherjee, S. S. Bhattacharya, A. Roy, \emph{Characterization of the quantumness of unsteerable tripartite correlations}, \href{https://arxiv.org/abs/1706.08415}{arXiv:1706.08415 [quant-ph] (2017).}
\bibitem{allmult} R. F. Werner, and M. M. Wolf, \emph{All-multipartite Bell-correlation inequalities for two dichotomic observables per site}, \href{https://journals.aps.org/pra/abstract/10.1103/PhysRevA.64.032112}{Phys. Rev. A \textbf{64}, 032112 (2001).}
\bibitem{TGS05} G. Toth, O. Guhne, M. Seevinck, J. Uffink, \emph{Addendum to ``Sufficient conditions for three-particle entanglement and their tests in recent experiments"}, \href{https://journals.aps.org/pra/abstract/10.1103/PhysRevA.72.014101}{Phys. Rev. A \textbf{72}, 014101 (2005).}
\bibitem{KZ96} G. Krenn and A. Zeilinger, \emph{Entangled entanglement}, \href{https://journals.aps.org/pra/abstract/10.1103/PhysRevA.54.1793}{Phys. Rev. A \textbf{54}, 1793  (1996).}
%\bibitem{BPR+17} S. S. Bhattacharya, B. Paul, A. Roy, A. Mukherjee, C. Jebaratnam, and M. Banik, Phys. Rev. A \textbf{95}, 042130 (2017).



%\bibitem{GBG+11} G. L. Giorgi, B. Bellomo, F. Galve, and R. Zambrini, Phys. Rev. Lett. 107, 190501 (2011).
%\bibitem{MGO+16} T. Moroder, O. Gittsovich, M. Huber, R. Uola, and O. Guhne, Phys. Rev. Lett. \textbf{116}, 090403 (2016).
\end{thebibliography}
\end{document}